\documentclass[11pt]{article}

\usepackage{tikz}

\author{Nadiia Chepurko \\ MIT \\ \href{mailto:nadiia@mit.edu}{\texttt{nadiia@mit.edu}} \and Kenneth L. Clarkson \\ IBM Research \\ \href{mailto:klclarks@us.ibm.com}{\texttt{klclarks@us.ibm.com}} 

\and
Lior Horesh \\ IBM Research \\ \href{mailto:lhoresh@us.ibm.com}{\texttt{lhoresh@us.ibm.com}} 

\and Honghao Lin \\ CMU \\ \href{mailto:honghaol@andrew.cmu.edu}{\texttt{honghaol@andrew.cmu.edu}}

\and David P. Woodruff \\ CMU \\ \href{mailto:dwoodruf@cs.cmu.edu}{\texttt{dwoodruf@cs.cmu.edu}}}
\date{}
\title{Quantum-Inspired Algorithms from Randomized Numerical Linear Algebra}

\usepackage{amsmath,amssymb,amsthm,amsfonts,latexsym,bbm,xspace,graphicx,float,mathtools,mathdots}
\usepackage{braket,caption,subcaption,ellipsis,xcolor,textcomp,hhline,pifont,combelow}
\usepackage[colorlinks,citecolor=blue,bookmarks=true,pagebackref=true, urlcolor=blue]{hyperref}
\usepackage[nameinlink]{cleveref}
\usepackage[letterpaper,margin=1in]{geometry}
\usepackage{enumitem}
\usepackage{inconsolata}

\usepackage{url}
\usepackage{amsmath}
\usepackage{mathtools}
\usepackage{amssymb}
\usepackage{graphicx}
\usepackage{stmaryrd}
\usepackage{algorithm}
\usepackage{algorithmic}[section]
\usepackage{arydshln}
\usepackage{multirow}

\newcommand{\norm}[1]{{\| #1 \|}}

\def\SS{{S}}

\def\hL{{\hat{L}}}

\def\rank{\textsf{rank}}

\def\rows{n}
\def\colms{d}
\def\lognd{\log(\rows\colms)}

\newcommand\twomat[2]{\left[\begin{smallmatrix} #1 \\ #2 \end{smallmatrix} \right] }

\newcommand\Al{A_{(\lambda)}}
\newcommand\hAl{\hat{A}_{(\lambda)}}
\newcommand\Iden{{I}}
\newcommand{\tA}{{\hat A}}

\newcommand{\cE}{\mathcal{E}}

\newcommand{\hA}{{\hat A}}
\newcommand{\ha}{{\hat a}}

\newcommand{\rA}{k}

\newcommand{\tO}{\tilde {O}}
\newcommand{\tX}{\tilde {X}}
\newcommand{\tq}{\tilde {q}}

\newcommand{\hb}{\hat{b}}
\newcommand{\hB}{\hat{B}}
\newcommand{\hU}{\hat{U}}
\newcommand{\hSigma}{\hat{\Sigma}}
\newcommand{\hV}{\hat{V}}
\newcommand{\tY}{\tilde{Y}}

\DeclareMathOperator{\nnz}{\mathtt{nnz}}
\DeclareMathOperator{\E}{{E}}
\DeclareMathOperator{\argmin}{argmin}

\DeclareMathOperator{\colspan}{\mathrm{im}}

\DeclareMathOperator{\Samp}{\textsc{Samp}}
\DeclareMathOperator{\Sampler}{\textsc{Sampler}}
\DeclareMathOperator{\DynSampler}{\textsc{DynSampler}}

\DeclareMathOperator{\DynSamp}{\textsc{DynSamp}}

\def\samplesize{v}

\newcommand\sd{{\texttt{sd}}}

\newcommand\tx{{\tilde x}}
\newcommand\hsig{{\hat\sigma}}
\newcommand\hkap{{\hat\kappa}}

\newcommand{\defeq}{\stackrel{\textit{\tiny{def}}}{=}}
\newcommand{\poly}{{\mathrm{poly}}}
\newcommand{\eps}{\varepsilon}
\newcommand{\R}{{\mathbb R}}

\newtheorem{theorem}{Theorem}

\newtheorem{definition}[theorem]{Definition}
\newtheorem{lemma}[theorem]{Lemma}

\newtheorem{remk}[theorem]{Remark}
\newtheorem{exmp}[theorem]{Example}

\newtheorem{problem}[theorem]{Problem}

\renewcommand{\epsilon}{\varepsilon}
\begin{document}
\begin{titlepage}
\maketitle
\thispagestyle{empty}

\begin{abstract}
We create classical (non-quantum) dynamic data structures supporting queries for recommender systems and least-squares regression that are comparable to their quantum analogues. De-quantizing such algorithms has received a flurry of attention in recent years; we obtain sharper bounds for these problems. More significantly, we achieve these improvements by arguing that the previous quantum-inspired algorithms for these problems are doing leverage or ridge-leverage score sampling in disguise;
these are powerful and standard techniques in randomized numerical linear algebra. With this recognition, we are able to employ the large body of work in numerical linear algebra to obtain algorithms for these problems that are simpler or faster (or both) than existing approaches. Our experiments demonstrate that the proposed data structures also work well on real-world datasets.
\end{abstract}

\end{titlepage}

\vspace{-0.1in}
\section{Introduction}
\vspace{-.1in}
In recent years, quantum algorithms for various problems in numerical linear algebra have been
proposed, with applications including least-squares regression and recommender systems \cite{harrow2009quantum,lloyd2016quantum, PhysRevLett.113.130503,gilyen2019quantum, zhaoquantumgaussian,brandoQuantumSdp, vanapeldoornImprovedQuantumSdp,lloyd2014quantum, cong2016quantum,BerryHamSim}.
Some of these algorithms have the striking property that their running times do not depend on
the input size. That is, for a given matrix $A\in\R^{\rows\times\colms}$ with $\nnz(A)$ nonzero entries,
the running times for these proposed quantum algorithms are at most polylogarithmic in 
%$\nnz(A)$, 
$\rows$ and $\colms$, 
and polynomial in other parameters of $A$,
such as $\rank(A)$, the condition number $\kappa(A)$, or Frobenius norm $\|A\|_F$. 
%\klcin{changed "sublinear in nnz(A)" to "polylog in nnz(A)"; are there examples of the former that are not the latter?}

However, as observed by Tang \cite{tang2019quantum} and others, there is a catch:
these quantum algorithms depend on a particular input representation of $A$, which is
a simple data structure that allows $A$ to be employed for quick preparation of a quantum state suitable
for further quantum computations. This data structure, which is a collection of weighted complete binary
trees, also supports rapid weighted random sampling of~$A$, for example, sampling
the rows of~$A$ with probability proportional to their squared Euclidean lengths. 
So, if an ``apples to apples''
comparison of quantum to classical computation is to be made, it is reasonable
to ask what can be
accomplished in the classical realm using the sampling that the given
data structure supports. 

A recent line of work analyzes the speedups of these quantum algorithms by developing classical counterparts that exploit these restrictive input and output assumptions, and shows that previous quantum algorithms do not give an exponential speedup.
In this setting, it has recently been shown
that sublinear time is sufficient for least-squares regression
using a low-rank design matrix $A$ \cite{gilyen2018quantum, DBLP:journals/corr/abs-1811-04852}, for computing a low-rank
approximation to input matrix $A$ \cite{tang2019quantum}, 
and for solving ridge regression problems \cite{gilyen2020improved},
using classical (non-quantum) methods, assuming the data structure of trees
has already been constructed. Further, the results obtained in  \cite{tang2019quantum, gilyen2018quantum, gilyen2020improved} serve
as appropriate comparisons of the power of quantum to classical computing,
%and also in the context of classical computing, 
due to their novel input-output model:
data structures are input, then sublinear-time computations are done,
%\Lior{'computation' instead of 'work'},
yielding data structures as output.

The simple weighted-sampling
data structure used in these works to represent the input can be efficiently constructed and stored: 
it uses $O(\nnz(A))$ space, with a small constant overhead, and requires
$O(\nnz(A))$ time to construct, in the static case where the matrix $A$ is given
in its entirety,
and can support updates and queries to individual entries of $A$ in $O(\log (nd))$ time.
However, the existing reported sublinear bounds are high-degree polynomials
in the parameters involved: for instance, the sublinear term in the running time for low-rank least-squares regression is  $\tO(\rank(A)^6\norm{A}_F^6\kappa(A)^{16}/\eps^6)$;
see also more recent work for ridge regression \cite{gilyen2020improved}. This combination of features raises the following question:
\vspace{-0.1in}
\begin{quote}
% \centering
\emph{
Question 1: Can the sublinear terms in the running time be reduced significantly while preserving the leading order dependence of $O(\nnz(A))$ and $O(\lognd)$ per update (dynamic)? %and $O(\nnz(A)\log n)$ %time?
}
\end{quote}
\vspace{-0.1in}
%  (For the latter, assuming $O(\nnz(A))$ updates.)

Perhaps a question of greater importance is the connection between quantum-inspired algorithms and the vast body of work in randomized numerical linear algebra: see the surveys \cite{DBLP:journals/fttcs/KannanV09,DBLP:journals/ftml/Mahoney11,woodruff2014sketching}. There are a large number of randomized algorithms based on sampling and sketching techniques for problems in linear algebra. Yet prior to our work, none of the quantum-inspired algorithms, which are sampling-based, have discussed the connection to leverage scores, for example, which are a powerful and standard tool. 
%randomized numerical linear algebra. 
\vspace{-0.1in}
\begin{quote}
% \centering
\emph{Question 2: Can the large body of work in randomized numerical linear algebra be applied effectively in the setting of quantum-inspired algorithms?
}
\end{quote}
\vspace{-0.1in}

\subsection{Our Results} 
We answer both of the questions above affirmatively. In fact, we answer Question 1 by answering Question 2. Namely, we obtain significant improvements in the sublinear terms,
% for the dynamic (fast updates to the input matrix) and static settings,
and our analysis relies on simulating \textit{leverage score sampling} and \textit{ridge leverage score sampling},
using the aforementioned data structure to sample rows proportional to squared Euclidean norms. Additionally, we empirically demonstrate the speedup we achieve on real-world and synthetic datasets (see Section \ref{sec:exp}).

\paragraph{Connection to Classical Linear Algebra and Dynamic Data Structures.} The work on quantum-inspired algorithms builds data structures for sampling according to the squared row and column lengths of a matrix. This is also a common technique in randomized numerical linear algebra - see the recent survey on length-squared sampling by Kannan and Vempala \cite{kannan2017randomized}. However, it is well-known that leverage score sampling often gives
stronger guarantees than length-squared sampling; leverage score sampling was pioneered in the algorithms community in \cite{dmm06b}, and made efficient in \cite{drineas2012fast} (see also analogous prior work
in the $\ell_1$ setting, starting with \cite{clarkson2005subgradient}).
%\klcin{sorry, couldn't resist}

Given an $n \times d$ matrix $A$, with $n \geq d$, its (row) leverage scores are the squared row norms of $U$, where $U$ is an orthonormal basis with the same column span as $A$. One can show that any choice of basis gives the same scores. Writing $A = U \Sigma V^T$ in its thin singular value decomposition (SVD),
and letting $A_{i,*}$ and $U_{i,*}$ denote the $i$-th rows of $A$ and $U$ respectively,
we see that $\norm{A_{i,*}} = \norm{U_{i,*} \Sigma}$. Consequently,
letting $k=\rank(A)$, and with $\sigma_1$ and $\sigma_k$ denoting the maximum and minimum non-zero singular 
values of $A$, we have
$\norm{A_{i,*}} \geq \norm{U_{i,*}} \sigma_k(A),$ and $\norm{A_{i,*}} \leq \norm{U_{i,*}}\sigma_1(A)$.
%\klcin{Is it standard for "minimum singular value" to mean mininum nonzero sv?}

Thus, sampling according to the squared row norms of $A$ is equivalent to sampling from a distribution with ratio distance at most $\kappa^2(A) = \frac{\sigma_{1}(A)^2}{\sigma_{k}(A)^2}$ from the leverage score distribution, that is, sampling a row with probability proportional to its leverage score. This is crucial, as it implies using standard arguments (see, e.g., \cite{woodruff2014sketching} for a survey) that if we oversample by a factor of $\kappa^2(A)$, then we obtain the same guarantees for various problems that leverage score sampling achieves. Notice that the running times of quantum-inspired algorithms, e.g., the aforementioned $\tO(\rank(A)^6\norm{A}_F^6\kappa(A)^{16}/\eps^6)$ time for regression of \cite{gilyen2018quantum} and the $\tO(\norm{A}_F^8 \kappa(A)^2/(\sigma_{min}(A)^6 \eps^4))$ time for regression of \cite{gilyen2020improved}, both take a number of squared-length
samples of $A$ depending on $\kappa(A)$, and thus are implicitly doing leverage score sampling, or in the case of ridge regression, ridge leverage score sampling.

Given the connection above, we focus on two central problems in machine learning and numerical linear algebra,  ridge regression (Problem \ref{prob:ridge}) and low rank approximation (Problem \ref{prob:lra}). We show how to obtain simpler algorithms and analysis than those in the quantum-inspired literature by using existing approximate matrix product and subspace embedding guarantees of leverage score sampling. In addition to improved bounds, our analysis de-mystifies what the rather involved $\ell_2$-sampling arguments of quantum-inspired work are doing, and decreases the gap between quantum and classical algorithms for machine learning problems. We begin by formally defining ridge regression and low-rank approximation, and the dynamic data structure model we focus on.

\begin{problem}[Ridge Regression]
\label{prob:ridge}
Given an $n\times d$ matrix $A$, $n \times d'$ matrix $B$ and a ridge parameter $\lambda\geq 0$, the ridge regression problem is defined as follows: 
\[
\min_{X \in \R^{d \times d'} } \norm{ A X - B }_F^2 + \lambda \norm{X}_F^2,
\]
where $\norm{\cdot}_F^2$ denotes the sum-of-squares of entries.  
\end{problem}

\begin{problem}[Low-Rank Approximation]
\label{prob:lra}
Given an $n\times d$ matrix $A$ and a rank parameter $k \in [d]$, the low-rank approximation problem is defined as follows:  
\[\min_{X \in \R^{n \times d} : \rank(X)=k } \norm{ A - X }_F^2.
\]
\end{problem}

\begin{definition}[Dynamic Data Structure Model]
Given an $n\times d$ matrix $A$, the dynamic data structure supports the following operations in $O(\lognd)$ time: (a) sample row $A_{i,*}$ with probability $\norm{A_{i,*}}_2^2/\norm{A}^2_F$, (b) sample entry $j$ in row $i$ with probability $A_{i,j}^2/\norm{A_{i,*}}_2^2$ and (c) output the $(i,j)$-th entry of $A$. 
\end{definition}

We note that in this input model, reading the entire matrix  would be prohibitive and the algorithms can only access the matrix through the weighted sampling data structure.

We now describe our concrete results in more detail. At a high level, 
our algorithm for ridge regression, Algorithm~\ref{alg:dyn ridgereg}, does the following: sample a subset of the rows of $A$
via length-squared sampling, and take a length-squared sample of the columns of that subset. Then, solve a linear system
on the resulting small submatrix using the conjugate gradient method. Our analysis of this algorithm results in  the following theorem.
(Some of the (standard) matrix notation used is given in Section~\ref{subsec nota}.)

The residual error is bounded in the theorem using
$\frac1{\sqrt{2\lambda}}\norm{U_{\lambda,\perp}B}_F$, where $U_{\lambda,\perp}$
denotes the bottom
$m_S-p$ left singular vectors of a sketch $SA$ of $A$, where~$p$ is such that $\lambda$ is between
$\sigma^2_{p+1}(SA)$ and $\sigma^2_p(SA)$.
(Here we use $p=\rank{A}$ when $\lambda\le \sigma^2_{\rank{A}}$.)
We could write this roughly as $\norm{A_{-p}A^+_{-p}B}/\norm{A_{-p}}$,
where $A_{-p}$ denotes $A$ minus its best rank~$p$ approximation. It is the
part of $B$ we are ``giving up'' by including a ridge term. Proofs for this section are deferred to Appendix~\ref{sec:ridge_proof}.

\begin{theorem}[Dynamic Ridge Regression]\label{thm dyn ridgereg inf}
Given an $n \times d$ matrix $A$ of rank $k$,
%for which a sampling data structure has been maintained
an $n \times d'$ matrix $B$,
error parameter $\eps >0$, and ridge parameter $\lambda$, 
let $\kappa^2_\lambda= (\lambda+\sigma^2_1(A))/(\lambda + \sigma^2_k(A))$ be the ridge condition number of $A$, and let $\psi_\lambda=\norm{A}_F^2/(\lambda + \sigma_k^2(A))$. 
Further,  let 
$X^*$ be the optimal ridge regression solution, i.e., $X^* = \argmin_X \norm{AX-B}_F^2 +\lambda\norm{X}_F^2$. 

Then there is a data structure supporting turnstile updates of the entries of $A$ in $O(\lognd)$ time, and an algorithm using that data structure that computes
a sample $SA$ of $m= O\left( \kappa^2_\lambda\psi_\lambda\log(nd) / \epsilon^2 \right)$ rows of $A$, where $S\in\R^{m\times n}$ is a sampling matrix, and outputs
$\tX\in\R^{m\times d'}$, such that with probability $99/100$, 
\[
\norm{A^\top S^\top \tX - X^*}_F
    \le \eps\left(1 + 2\gamma \right) \norm{X^*}_F  + \frac{\eps}{\sqrt{\lambda}}\norm{U_{k,\perp}B}_F,
\]
where
$U_{\lambda,\perp}B$ is the projection of $B$ onto the subspace corresponding to the singular values of $SA$ less than $\sqrt\lambda$;
and $\gamma^2 = \frac{\norm{B}_F^2}{\norm{A A^+ X^*}_F^2}$ is a problem dependent parameter. 

Further, the running time of the algorithm is
$\tO\left(\colms'\eps^{-4}  \psi_\lambda^2 \kappa^2_\lambda \log(\colms)\right)$. 
Finally, for all $i\in[d]$, and $j\in[d']$, an entry $(A^\top S^\top \tX)_{i,j}$ can be computed in $O(m\lognd)$ 
%= O\left(\eps^{-2}\kappa^2_\lambda\psi_\lambda(\lognd)^2\right)$ 
time.
\end{theorem}

% \textcolor{red}{N:I cleaned up the the theorem statement above, let me know if it makes sense now. }

We note that the ``numerical'' quantities $\kappa_\lambda$  and $\psi_\lambda$ are decreasing in $\lambda$. 
% roughly at most $\norm{AA^+B}_F/\norm{A}$
%\nadiia{why is this true?}. \klc{omitted now.}
%\klcin{worth saying?}
When $\lambda$ is
within a constant factor of $\norm{A}^2$, 
%$\hkap$ is constant and 
$\psi_\lambda$ is within a constant factor of the \emph{stable rank}
$\norm{A}_F^2/\norm{A}^2$, where the stable rank is always at most $\rank(A)$. We also note that in the theorem, and the remainder of the paper, a \emph{row sampling matrix} $S$ has rows that are multiples of natural basis vectors,
so that $SA$ is a (weighted) sample of the rows of $A$. A column sampling matrix is defined similarly.

\paragraph{Concurrent Work on Ridge Regression.} In an independent and concurrent work, Gily{\'e}n, Song and Tang \cite{gilyen2020improved} obtain a roughly
comparable classical algorithm for regression, assuming access to the tree data structure, which runs in time $\tilde{O}\left( \frac{\| A\|^6_F \|A\|^2_2}{\|A^{+} \|^8_2 \eps^4 }\right)$,
or in the notation above, $\tO(\eps^{-4}\psi_\lambda^3\kappa^2)$, for the special case of $d'=1$. Their algorithm is based on Stochastic Gradient Descent.
%and the $\nnz(A)$ term in the running time of their algorithm does not get multiplied by $\poly(1/\eps)$ factors, despite the $O^*(\cdot)$ notation used in that paper.
%
%their claimed leading term ($\nnz(A)$) may incur a polynomial dependence on $1/\eps$ and polylogarithmic dependence on $n$ and $d$.
%\klcin{worth saying?}
%
%

Next, we describe our results for low-rank approximation.
We obtain a dynamic algorithm (Algorithm ~\ref{alg:dyn lowrank}) for approximating $A$ with a rank-$k$ matrix,
for a given $k$, and a data structure for sampling from it, 
in the vein of \cite{tang2019quantum}. At a high-level, as with our ridge regression
algorithm, we first sample rows proportional to their squared Euclidean norm (length-squared sampling) and then sample a subset of columns resulting in a small submatrix with $\tO(\eps^{-2}k)$
rows and columns. We then compute the SVD of this matrix, 
and then work back up to $A$ with more sampling and a QR factorization. The key component in our algorithm and analysis is using \emph{Projection-Cost Preserving} sketches (see Definition \ref{def PCP}). These enable us to preserve the Frobenius cost of projections onto all rank-$k$ subspaces simultaneously. As a result, we obtain the following theorem: 

\begin{theorem}[Sampling from a low-rank approximation]
\label{thm lra inf}
Given an $n \times d$ matrix $A$ for which a sampling data structure has been maintained,
target rank $k \in [d]$ and error parameter $\eps>0$, 
%,and estimates of numerical properties of $A$ (such as $\norm{A-A_{k}}_F^2$, the error of the best rank-$k$ approximation $A_{k}$),
we can find sampling matrices $S$ and $R$, and rank-$k$
matrix $W$, such that $\norm{ARWSA-A}_F \le (1+O(\eps))\norm{A-A_{k}}_F.$
Further, the running time is $\tO(\eps^{-6}k^3 + \eps^{-4} \psi_\lambda (\psi_\lambda + k^2 + k \psi_{k}))$, where $\psi_\lambda$ is as in Theorem~\ref{thm dyn ridgereg inf}, and $\psi_{k}  = \frac{\norm{A}_F^2}{\sigma_{k}(A)^2}$ .%\norm{A}_F^2 / \sigma_{k}(A)^2
Given $j\in [\colms]$, a random index $i\in [\rows]$ with probability distribution
 $(ARWSA)_{ij}^2/\norm{ARWSA)_{*,j}}^2$
 can be generated in expected time $\tO(\psi_\rA + k^2\eps^{-2} \kappa^2)$,
 where $\kappa= \sigma_1(A)\cdot \sigma_{\rank(A)}(A)$.
\end{theorem}

Here if the assumption $\norm{A_{k}}_F^2\ge \eps\norm{A}_F^2$
does not hold, the trivial solution $0$ satisfies the relative error target and we assume the resulting approximation is not worth sampling:
% \begin{align*}
% \norm{A-0}_F^2 & \le \frac1{1-\eps}(\norm{A}_F^2 - \norm{A_{k}}_F^2) \\
% &= \frac1{1-\eps}\norm{A-A_{k}}_F^2 \le (1+2\eps)\norm{A-A_{k}}^2,
% \end{align*}
\[
\norm{A-0}_F^2 \le \frac1{1-\eps}(\norm{A}_F^2 - \norm{A_{k}}_F^2) 
= \frac1{1-\eps}\norm{A-A_{k}}_F^2. %\le (1+2\eps)\norm{A-A_{k}}^2,
\]

% \begin{theorem}\label{thm dyn LRA}
%  Given $\DynSamp(A)$ for $A\in\R^{\rows\times \colms}$,
%  target rank $k$,
%  $\hsig_\rA\le 1/\norm{A^+}$,
%  $\hsig_{k} \le \sigma_{k}(A)$,
% error parameter $\eps$, and 
% estimate $\tau$ of $\norm{A-A_{k}}_F^2$.
% We assume $\norm{A_{k}}_F^2\ge \eps\norm{A}_F^2$.
% Then 
% \textsc{BuildLowRankFactors}, Algorithm~\ref{alg:dyn lowrank}, returns
% $W$ of rank $k$ and sampling matrices $S$ and $R$ such that
% $\norm{ARWSA-A}_F \le (1+O(\eps))\norm{A-A_{k}}_F$,
% where $A_{k}$ is the best rank-$k$ approximation
% to $A$. The time taken to find $W$, $S$, and $R$ is
% \[
% \tO(\eps^{-6}k^3 + \eps^{-4} \psi_\lambda (\psi_\lambda + k^2 + k \psi_{k})),
% \]
% where $\psi_\lambda= \norm{A}_F^2/(\tau/k + \hsig_\rA^2)$
% and $\psi_{k} = \norm{A}_F^2/\sigma_{k}(A)$.
% Given $j\in [\colms]$, $i\in [\rows]$ can be generated with probability
% $(ARWSA)_{ij}^2/\norm{ARWSA)_{*,j}}^2$ in expected time $O(\norm{A}_F^2/\hsig_\rA^2 + m_R^2 \kappa(A)^2)$.
% \end{theorem}

% \begin{theorem}[Sampling from LRA, informal Theorem \ref{thm rec} ]
% Given an $n \times d$, rank-$k$ matrix $A$, we can build a data structure of size $O(nnz(A) + k^3 + k^2/\eps)$ in time $O(\nnz(A) + k^4 + k/\eps^7)$ such that given $j \in[d]$, it outputs an $i \in [n]$ with probability proportional to $\hat{A}^2_{i,j}/\| \hat{A}_{*,j} \|^2_2$, where $\| A - \hat{A} \|^2_F \leq (1+\eps) \|A - A_k \|^2_F$. Further, the expected query time $O( (k/\eps + k^2) n^{0.1})$. 
% \end{theorem}

This result is directly comparable to Tang's algorithm  \cite{tang2019quantum} for recommender systems which again needs query time that is a large polynomial in $k, \kappa$ and $\eps^{-1}$. Our algorithm returns a relative error
approximation, a rank-$k$ approximation within $1+\eps$ of the best rank-$k$ approximation; Tang's algorithm
has additive error, with a bound more like
$\norm{A-A_{k}}_F + \eps\norm{A}_F$. Finally, we note that $\psi_k \leq \kappa^2$ and for several settings of $k$ can be significantly smaller. 

For ease of comparison we summarize our results in Table \ref{tab:table1}.

\begin{table*}[htb!]
\normalsize
\begin{center}
    \caption{Comparison of our results and prior work. Let the target error be $\eps$, target rank be $k$ and let
    $\psi_{k} = \norm{A}_F^2/\sigma_{k}(A)^2$ , where $\sigma_k$ is the $k$-th singular value of the input matrix. 
    %For $\lambda\geq0$, let 
    %$\psi_\rA = \norm{A}_F^2/\hsig_k^2$,
    %${\psi_\lambda=\norm{A}_F^2/(\lambda + \sigma_k^2)}$ and
    %$\kappa_\lambda^2 = (\lambda+\sigma_1^2(A))/(\lambda + \sigma_k^2(A))$.
    %$\hkap^2 = (\lambda+\hsig_1^2)/(\lambda + \hsig_k^2)$,
    Also, $\hsig_k\le 1/\norm{A^+}$, $\hsig_1 \ge \norm{A}$,
    $d'$ is the number of columns of $B$ for multiple-response, and 
    $\eta$ denotes some numerical properties of $A$. To avoid numerous parameters, we state our results by setting $\lambda = \Theta(\|A\|_2^2)$ in the corresponding theorems.
    % above
    } %\Lior{we need to predefine $\lambda_r$, $\lambda_s$}
    \label{tab:table1}
\begin{tabular}{|l|c|l|c|l|}
\hline
Problem & \multicolumn{2}{c|}{Time} & \multicolumn{2}{c|}{Prior Work} 
\\ \hline
& Update & \multicolumn{1}{c|}{Query}   & Update & \multicolumn{1}{c|}{Query} 
\\ \hline
% \begin{tabular}[c]{@{}l@{}}Leverage Score\\ Sampling\end{tabular} & \multicolumn{1}{l|}{$O(\log(n))$} & $\tO(k^{(1.01)\omega} + \samplesize k^{2.01} (\kappa \log(\rows))^2)$ & \multicolumn{1}{l|}{NA} & NA         
% \\ \hline
\begin{tabular}[c]{@{}l@{}}Ridge\\ Regression\end{tabular} & \multicolumn{1}{l|}{$O(\log(n)$)} & $\tO\left( \frac{\colms'   \kappa^3 \norm{A}_F^2 \log(\colms) }{\eps^{4} \norm{A}_2^2 }   \right)$  & \multicolumn{1}{l|}{$O(\log(n))$} &   $\tO\left(\frac{k^6\norm{A}_F^6\kappa^{16}}{\eps^6}\right)$ \\
&&Thm. \ref{thm dyn ridgereg inf}&&\cite{gilyen2018quantum} \\
&&&& $\tO\left(\frac{\norm{A}_F^8 \kappa(A)^2}{(\sigma_{\min}^6 \eps^4)}\right)$\\
&&&&\cite{gilyen2020improved}
\\ \hline
\begin{tabular}[c]{@{}l@{}}Low Rank\\ Sampling\end{tabular} & \multicolumn{1}{l|}{$O(\log(n))$} & $\tO\left( \frac{ \norm{A}_F^2 \left( \frac{\norm{A}_F^2}{\norm{A}_2^2} + k^2 + k \psi_{k}\right)}{\eps^{4} \norm{A}_2^2}+\frac{k^3}{\eps^{6}} \right)$ & \multicolumn{1}{l|}{$O{(\log(n))}$} & $\Omega( \textrm{poly}(\kappa k \eps^{-1}\eta) )$ \\
&& Thm.~\ref{thm lra inf}&&\cite{tang2019quantum}
\\ \hline
\end{tabular}
\end{center}
\end{table*}

% \vspace{-0.1in}
\subsection{Related Work}
% \vspace{-0.1in}
\paragraph{Matrix Sketching.}
The \textit{sketch and solve} paradigm \cite{clarkson2015input,woodruff2014sketching} was designed to reduce the dimensionality of a problem, while maintaining enough structure such that a solution to the smaller problem remains an approximate solution the original one. This approach has been pivotal in speeding up basic linear algebra primitives such as least-squares regression \cite{s06, rokhlin2008fast, clarkson2015input}, $\ell_p$ regression \cite{cohen2015lewis, wang2019tight}, low-rank approximation \cite{NN13, cohen2017input, li2020input}, linear and semi-definite programming \cite{cohen2019solving,jiang2020inv, jiang2020faster} and solving non-convex optimization problems such as $\ell_p$ low-rank approximation \cite{song2017low,song2019relative,ban2019ptas} and training neural networks \cite{bakshi2019learning,brand2020training}. For a comprehensive overview we refer the reader to the aforementioned papers and citations therein.  Several applications use rank computation, finding a full rank subset of rows/columns, leverage score sampling, and computing subspace embeddings as key algorithmic primitives.
% In addition to being used as a block box, we believe our techniques will be useful in sharpening bounds for several such applications.

\paragraph{Sublinear Algorithms and Quantum Linear Algebra.}
Recently, there has been a flurry of work on sublinear time algorithms for structured linear algebra problems \cite{musco2017sublinear, shi2019sublinear, balcan2019testing, bakshi2020testing}  and quantum linear algebra \cite{harrow2009quantum, gilyen2019quantum,lloyd2014quantum, kerenidis2016quantum, dunjko2020non}. The unifying goal of these works is to avoid reading the entire input to solve tasks such as linear system solving, regression and low-rank approximation. The work on sublinear algorithms assumes the input is drawn from special classes of matrices, such as positive semi-definite matrices \cite{musco2017sublinear, bakshi2019robust}, distance matrices \cite{bakshi2018sublinear, indyk2019sample} and Toeplitz matrices \cite{lawrence2020low}, whereas the quantum algorithms (and their de-quantized analogues) assume access to data structures that admit efficient sampling \cite{tang2019quantum, gilyen2018quantum,chia2020sampling}.

The work  of Gily{\'e}n, Lloyd  and Tang \cite{gilyen2018quantum} on low-rank least squares
produces a data structure as output: given index $i\in [\colms]= \{1, \ldots,\colms\}$,
the data structure returns entry $x'_i$ of $x'\in\R^\colms$,
which is an approximation to
the solution $x^*$ of $\displaystyle \min_{x\in\R^\colms}\norm{Ax-b}$, where $b\in\R^\rows$.
The error bound is $\norm{x'-x^*}\le \eps\norm{x^*}$, for given $\eps>0$.
This requires the
condition that $\norm{Ax^*-b}/\norm{Ax^*}$ is bounded above by a constant. Subsequent work \cite{chia2020sampling} removes this requirement, and both results obtain data structures that need space polynomial in $\rank(A)$, $\eps$, $\kappa(A)$,\footnote{Throughout, we define $\kappa(A)=\norm{A}\norm{A^+}$,
that is, the ratio of largest to smallest \emph{nonzero} singular values of $A$, so that, in particular,
it will never be infinite or undefined.}
and other parameters.

The work \cite{tang2019quantum} also produces a data structure,
that supports sampling relevant to the setting of recommender systems:
the nonzero entries of the input matrix $A$ are a subset of the entries of a
matrix $P$ of, for example, user preferences.
An entry $A_{ij}\in [0,1]$ is one if user $j$ strongly prefers product $i$,
and zero if user $j$ definitely does not like product~$i$. It is assumed that $P$ is
well-approximated by a matrix of some small rank~$k$. The goal is to estimate $P$ using
$A$; one way to make that estimate effective, without simply returning all entries of $P$,
is to create a data structure so that given $j$, a random index $i$ is returned,
where $i$ is returned with probability $\ha_{ij}^2/\norm{\hat A_{*,j}}^2$. Here $\hA_{*,j}$ is the $j$'th column
 of $\hA$ (and $\ha_{ij}$ an entry), where
$\hat A$ is a good rank-$k$ approximation to $A$, and therefore,
under appropriate assumptions, to $P$. The estimate $\hat A$ 
is regarded as a good approximation %\Lior{approximation}
if $\norm{\hat A - A}_F \le (1+\eps)\norm{A-[A]_k}_F$, where $[A]_k$ is the
matrix of rank $k$ closest to $A$ in Frobenius norm. Here $\eps$ is a given error parameter.
As shown in
\cite{tang2019quantum}, this condition (or indeed, a weaker one)
implies that the described sampler is useful in the context
of recommender systems.

\section{Outline}

The next section gives some notation and mathematical preliminaries, in particular
regarding leverage-score and length-squared sampling. This is followed by descriptions of our data structures and algorithms, and then by our computational experiments. The appendices give some extensive descriptions, proofs of theorems, and in Appendix~\ref{sec more exp}, some additional experiments.

\vspace{-0.1in}
\section{Preliminaries}
\label{subsec nota}
% \vspace{-.1in}
Let $X^+$ denote the Moore-Penrose pseudo-inverse of matrix $X$, equal to $V\Sigma^{-1}U^\top$ when
$X$ has thin SVD $X=U\Sigma V^\top$, so that $\Sigma$ is a square invertible matrix. 
%\Lior{it looks like you assume X to be full ranked, yet, why, would a pseudo inverse be pursued, rather than an actual inverse in such settings ? or do I miss anything about $\Sigma$ being square invertible}
%\Ken{$X$ can be rectangular, by ``thin'' I mean that the $U$, $\Sigma$, and $V$ have $\rank(A)$ columns, and $\Sigma$ is square with $\rank(A)$ nonzero singular values}
We note that $X^+ = (X^\top X)^+ X^\top = X^\top(XX^\top)^+ \text{ and } X^+XX^\top = X^\top,$
which is provable using the SVDs of $X$ and $X^+$. Also, if $X$ has full column rank, so that $V$ is square,
then $X^+$ is a left inverse
of $X$, that is, $X^+X = I_d$, where $d$ is the number of columns of  $X$.
Let $\norm{X}$ denote the spectral (operator) norm of $X$.
Let $\kappa(X)=\norm{X^+}\norm{X}$ denote the condition number of $X$.
We write $a\pm b$ to denote the set $\{c \mid |c-a|\le |b|\}$, and $c=a\pm b$ to denote the condition
that $c$ is in the set $a\pm b$.
Let $[m]=\{1,2,\ldots,m\}$ for an integer $m$. 

As mentioned, $\nnz(A)$ is the number of nonzero entries of $A$, and we assume $\mathrm{nnz}(A) \geq n$, which can be ensured by removing any rows of $A$ that only contain zeros. 
We let $[A]_k$ or sometimes $A_k$ denote the best rank-$k$ approximation to~$A$.
Let $0_{a\times b}\in\R^{a\times b}$ have all entries equal to zero, and similarly $0_a\in\R^a$ denotes the zero vector. Further, for an $n \times d$ matrix $A$ and a subset $S$ of $[n]$, we use the notation $A_{\mid S}$ to denote the restriction of the rows of $A$ to the subset indexed by $S$.
%We use Knuth's notation $x \uparrow \uparrow y$ to denote a tower of $x$'s, exponentiated $y$ times. 
As mentioned,
$n^\omega$ is the time needed to multiply two $n\times n$
matrices. 

%We use $\alpha(n)$ to denote the inverse Ackermann function.

% We will use length-squared sampling to obtain subspace embeddings. The analysis is shown in Appendix~\ref{sec:appendix_prelim}.

\begin{lemma}[Oblivious Subspace Embedding Theorem 7.4~\cite{chepurko2022near}]\label{lem embed}
For given matrix $A\in\R^{\rows\times\colms}$ with $\rA=\rank(A)$, there exists an oblivious sketching matrix $S$ that samples $m = O(\eps_0^{-2}\rA\log \rA)$ rows of $A$ such that with probability at least $99/100$, for all $x \in \mathbb{R}^d$,
$S$ is an $\eps_0$-subspace embedding, that is,
$\norm{SAx} = (1\pm\eps_0)\norm{Ax}$. Further, the matrix $SA$ can be computed in $O(\nnz(A)+k^\omega \textrm{poly}(\log \log (k) )+ \textrm{poly}(1/\eps_0)k^{2+o(1)})$ time.
\end{lemma}

% \textcolor{red}{N:Restated the above lemma.}

We obtain the following data structure for leverage-score sampling. We provide a statement of its properties below, but defer the description of the algorithm and proof to the supplementary material. While leverage-score sampling is well-known, we give an algorithm for completeness; also, our algorithm removes a $\log$ factor in some terms in the runtime, due to our use of the sketch of Lemma~\ref{lem embed}.
%\klcin{citep our arxiv paper?}

\begin{theorem}[Leverage Score Data Structure]\label{prelim:thm levsample}
Let $k=\rank(A)$, and choose $\mu_s \ge 1$.
Then, Algorithm~\ref{prelim_alg:levsample} ($\textsc{LevSample}(A,\mu_s, v)$)
 uses space $O(\rows + \rA^\omega\log\log(\rows\colms))$,
 not counting the space to store $A$, and runs in time
\[
O(\mu_s\nnz(A) + \rA^\omega \poly(\log\log(k)) + \rA^{2+o(1)} +  \samplesize k \rows^{1/\mu_s}),
\]
and outputs a leverage score sketching matrix $L$, which samples $v$ rows of $A$ with probability proportional to their leverage scores. (It also outputs a column selector $\Lambda$,
selecting an orthogonal basis of the column space of $A$.)
% For $\samplesize=O(\eps^{-2}k\log k)$, this bound can be expressed as
% $O(\mu_e\nnz(A) + \rA^\omega \textrm{poly}(\log\log(n)) + \eps^{-2-1/\mu_e}k^2)$ time, for $\mu_e\ge 1$.
\end{theorem}

% \textcolor{red}{N:Adding proof of this theorem to the appendix.}

\begin{definition}[Ridge Leverage-score Sample, Statistical Dimension]
Let $A$ be such that $\rA=\rank(A)$, and suppose $A$ has thin SVD $A=U\Sigma V^\top$, implying $\Sigma\in\R^{\rA\times\rA}$.
 For $\lambda>0$, let $\Al = \twomat{A}{\sqrt{\lambda} VV^\top}$ and $\Al$ has SVD $\Al = \twomat{U\Sigma D}{\sqrt{\lambda}VD}D^{-1}V^\top$, where
$D= (\Sigma^2 + \lambda\Iden_\rA)^{-1/2}$.
 Call ${\cal S} \subset [\rows]$
 a \emph{ridge leverage-score sample} of $A$ if each $i\in\cal S$ is chosen independently
 with probability at least $\norm{U_{i,*}\Sigma D}^2/\sd_\lambda(A)$, where the
 statistical dimension
 $\sd_\lambda(A) = \norm{U\Sigma D}_F^2 = \sum_{i\in [d]} \sigma_i^2/(\lambda+\sigma_i^2)$,
  recalling that $U\Sigma D$ comprises the top $n$ rows of the left singular matrix of $\Al$.
\end{definition}

We can also use \emph{length-squared sampling} to obtain subspace embeddings. In Section~\ref{subsec dyn lev}
we will give a data structure and algorithm that implements length-squared sampling. We defer the analysis to Appendix~\ref{sec:appendix_prelim}.

\begin{definition}[Length-squared sample]\label{def lensq sample}
 Let $A\in\R^{\rows\times\colms}$, $\lambda\geq0$, and $\Al$ be as in Lemma~\ref{lem Al}.
 For given $m$,
 let matrix $L\in\R^{m\times \rows}$ be chosen by picking each row of $L$ to be 
 $e_i^\top/\sqrt{p_i m}$, where $e_i\in\R^{\rows}$ is the $i$'th standard basis vector,
 and picking $i\in [\rows]$ with probability $p_i\gets \norm{A_{i,*}}^2/\norm{A}_F^2$.
\end{definition}

We obtain the corresponding lemma for length-squared sampling and defer the proof to Appendix \ref{sec:appendix_prelim}. 

\begin{lemma}[Length-squared sketch]\label{prelim:lem Lsquare}
Given a matrix $A \in \mathbb{R}^{n \times d}$ and a sample size parameter $\samplesize \in [n]$,  let $m = O\left(\samplesize\norm{\Al^+}^2\norm{A}_F^2/\sd_\lambda(A)\right)$. Then, with probability at least $99/100$, 
 the set of $m$ length-squared samples contains a ridge leverage-score sample
 of $A$ of size $\samplesize$.
\end{lemma}
\section{Dynamic Data Structures for Ridge Regression}
\label{subsec dyn lev}

In this section, we describe our dynamic data structures, and then our algorithm for solving Ridge Regression problems.
Given an input matrix $A\in\R^{\rows\times\colms}$, our data structure can be maintained under insertions and deletions (and changes)
in $O(\log (nd))$ time, such that sampling a row or column with probability proportional to its squared length can be done in $O(\lognd)$ time. The data structure is used for
solving both ridge regression and LRA (Low-Rank Approximation) problems. 

% The residual error is bounded in the theorem using
% $\frac1{\sqrt{2\lambda}}\norm{U_{\lambda,\perp}B}_F$, where $U_{\lambda,\perp}$
% denotes the bottom
% $m_S-p$ left singular vectors of a sketch $SA$ of $A$, where~$p$ is such that $\lambda$ is between
% $\sigma^2_{p+1}(SA)$ and $\sigma^2_p(SA)$.
% (Here we use $p=\rank{A}$ when $\lambda\le \sigma^2_{\rank{A}}$.)
% We could write this roughly as $\norm{A_{-p}A^+_{-p}B}/\norm{A_{-p}}$,
% where $A_{-p}$ denotes $A$ minus its best rank~$p$ approximation. It is the
% part of $B$ we are ``giving up'' by including a ridge term. Proofs for this section are deferred to Appendix~\ref{sec:ridge_proof}. 

First, we start with a simple folklore data structure.

\begin{lemma}\label{lem weighted tree}
Given $\ell$ real values $\{ u_i \}_{i \in [\ell]}$, there is a data structure
using storage $O(\ell)$, so that $L=\sum_{i \in [\ell]} u_{i}^2$ can be maintained, and
such that a random $i$ can be chosen with probability
$u_i^2/L$ in time $O(\log \ell)$. Values
can be inserted, deleted, or changed in the data structure in time $O(\log\ell)$.
\end{lemma}

The implementation of this data structure is discussed in Appendix~\ref{sec:ridge_proof}.
We use it in our data structure $\DynSamp(A)$, given below, which is used in
$\textsc{LenSqSample}$, Alg.~\ref{alg:lensqsample}, to sample rows and columns of $A$.

\begin{definition}\label{def dynsamp}
$\DynSamp(A)$ is a data structure that, for $A\in\R^{\rows\times\colms}$,
%and $\rAx$ bounding $\rank(A)$, 
comprises:
\begin{itemize}
  \setlength{\itemsep}{0ex}  
  \setlength{\parsep}{0ex}  
  \setlength{\parskip}{0ex}  
 \item For each row of $A$,
  the data structure of Lemma~\ref{lem weighted tree} for the nonzero entries of the row or column.
 \item For the rows of $A$, the data structure of Lemma~\ref{lem weighted tree}
 for their lengths.
 \item For given $i,j$, a data structure supporting access to the value of entry $a_{ij}$ of $A$ in $O(1)$ time.
\end{itemize}
\end{definition}

\begin{algorithm}[htb!]
\caption{$\textsc{LenSqSample}(\textsc{DS},S=\mathbf{null}, m_S, m_R)$}
\label{alg:lensqsample}
{\bf Input:}
	$\textsc{DS}=\DynSamp(A)$  (Def.~\ref{def dynsamp}) for $A\in\R^{\rows\times\colms}$, sample sizes $m_S$, $m_R$

{\bf Output:} Sampling matrices ${S\in\R^{m_S\times \rows}, R\in\R^{\colms\times m_R}}$\\
\begin{algorithmic}[1]
\STATE if $S==\mathbf{null}$\\
   $\qquad$ Use $\textsc{DS}$ to build row sampler $S\in\R^{m_S\times \rows}$ of $A$
\STATE Use $\textsc{DS}$ and $S$ to build column sampler $\R^{\colms\times m_R}$ of $SA$
\\ \COMMENT{cf. Lemma~\ref{lem dynsamp}}
\STATE return $S,R$
\end{algorithmic}
\end{algorithm}

% \begin{algorithm}[htb!]
% \caption{$\textsc{LenSqEmbed}(\DynSampler, Z, 
% \\ \texttt{dim}=\texttt{row}, \hat k = \texttt{unknown}, \eps = \eps_0, \samplesize=0)$}
% \label{alg:lensqembed}
% {\bf Input:}
% 	$\DynSampler=\DynSamp(A)$  (Def.~\ref{def dynsamp}) for $A\in\R^{\rows\times\colms}$
% 	$Z$ with $Z\ge\norm{A^+}$,
% 	\texttt{dim} flags sampling rows or columns (default is \texttt{row}),
% 	$\hat k=\rank(A)$ if not \texttt{unknown},
% 	$\eps$ error parameter (defaults to $\eps_0$, a small constant),
% 	$\samplesize$ leverage-score sample size, and $m_\hL = O(\eps^{-2} Z^2 \norm{A}_F^2\log\colms)$. \\
% {\bf Output:} $\hL\in\R^{m_\hL\times \rows}$
% \begin{algorithmic}[1]
% % \STATE if \texttt{dim}==\texttt{col}\\
% %   $\qquad$ $\AA\gets\AA^\top$
% % \STATE if $\hat k ==\texttt{unknown}$\\
% % $\qquad m_\hL = O(\eps^{-2} Z^2 \norm{A}_F^2\log\colms)$\\
% % else\\
% % $\qquad m_\hL = O(\frac{v}{\hat k}Z^2 \norm{A}_F^2)$
% \STATE Construct $\hL$ from  $m_\hL$ samples of the rows of $A$ as in Def.~\ref{def lensq sample} and $\DynSampler$, and using $w$, adding row $e_i^\top/\sqrt{m_\hL p_i}$ to $\hL$ when the sampling chose row $i$.
% \STATE return $\hL$
% \end{algorithmic}
% \end{algorithm}

\begin{algorithm}[htb!]
 {\small
\caption{$\textsc{RidgeRegDyn}(\textsc{DS}, B, \hsig_\rA, \hsig_1, \eps, \lambda)$}
\label{alg:dyn ridgereg}
{\bf Input:} $\textsc{DS}=\DynSamp(A)$, $B\in\R^{\rows\times\colms'}$, $\hsig_\rA\le 1/\norm{A^+}$, $\hsig_1 \ge \norm{A}$,  $\eps$ an error parameter, $\lambda$ a ridge weight\\
{\bf Output:} Data for approximate ridge regression solution $A^\top S^\top\tX$ where $S$ is a sampling matrix\\
\begin{algorithmic}[1]
\STATE $Z_\lambda \gets 1/\sqrt{\lambda + \hsig_\rA^2}$, $\hkap\gets Z_\lambda\sqrt{\lambda + \hsig_1^2} $
\STATE Choose $m_S = O(\eps^{-2}\hkap^2Z_\lambda^2 \norm{A}_F^2\log(\colms))$, \\
    $\quad$ $m_R=O(\hat m_R Z_\lambda^2 \norm{A}_F^2)$, where $\hat m_R = O(\eps^{-2}\log m_S)$
% , $Z\gets 1/\hsig_\rA$
% \STATE $S \gets \textsc{LenSqEmbed}(A, \DynSampler, Z_\lambda, \texttt{row}, \texttt{unknown},$
% \\$  \eps/\hkap)$
%  \   \COMMENT{Alg.~\ref{alg:lensqembed}; $m_S = O(\eps^{-2}\hkap^2Z_\lambda^2 \norm{A}_F^2\log(\colms))$ rows}
% \STATE $S'' \gets \textsc{LenSqEmbed}(A, \DynSampler, Z_\lambda, \texttt{row}, \texttt{unknown})$
% \\ $\qquad$ \COMMENT{Resulting $S''$ has $m_{S''} = O(Z^2 \norm{A}_F^2\log d)$, for a subspace $\eps_0$-embedding}
% \STATE $S\gets\twomat{S'}{S''}$
% \STATE $R \gets \textsc{LenSqEmbed}(SA, \DynSamp(A), Z_\lambda, \texttt{col}, m_S, \eps,$
% \\ $m_S\hat m_R)$ 
\STATE $S,R\gets \textsc{LenSqSample}(\textsc{DS}, \mathbf{null}, m_S, m_R)$ \COMMENT{cf. Alg.~\ref{alg:lensqsample};}
%  $\hat m_R = O(\eps^{-2}\log m_S)$, $m_R=O(\hat m_R Z_\lambda^2 \norm{A}_F^2)
% $}
\STATE $\tX \gets (SARR^\top A^\top S^\top + \lambda \Iden_{m_S})^{-1} SB$
\\ $\quad$ \COMMENT{Solve using conjugate gradient}
\STATE return $\tX$, $S$
\\ $\quad$ \COMMENT{approximate ridge regression solution is $ A^\top S^\top \tX$}
\end{algorithmic}
}% small
\end{algorithm}

\begin{lemma}\label{lem dynsamp}
 $\DynSamp(A)$ can be maintained under turnstile updates of $A$ in $O(\lognd)$ time.
  Using $\DynSamp(A)$, rows can be chosen
  at random with row $i\in [\rows]$ chosen with probability $\norm{A_{i,*}}^2/\norm{A}_F^2$
  in $O(\lognd)$ time.
  
  If $S\in\R^{m\times \rows}$ is a sampling matrix, so that $SA$ has rows that are each a
  multiple of a row of $A$, then
  $c$ columns can be sampled from $SA$ using $\DynSamp(A)$ in $O((c + m)\log(nd))$ time,
  with the column $j\in [\colms]$ chosen with probability $\norm{(SA)_{*,j}}^2/\norm{SA}_F^2$.
 \end{lemma}

% We will regard $\DynSamp(A^\top)$ and $\DynSamp(A)$ as interchangeable, and use these observations for sampling rows of sampled submatrices $AR$ and columns of $LA$.

We designate the algorithm of Lemma~\ref{prelim:lem Lsquare} as $\textsc{LenSqSample}$,
as given at a high level in Algorithm~\ref{alg:lensqsample}, and in more detail in the
proof of Lemma~\ref{lem dynsamp} in Appendix~\ref{sec:ridge_proof}.
% It will be used in two ways: when the rank is unknown, a sample large enough to be an $\eps$-embedding is returned,
% otherwise, a sample of size so that a leverage score sample of size $\samplesize$ is expected. 

This simple data structure and sampling scheme will be used to solve
ridge regression problems, via Algorithm~\ref{alg:dyn ridgereg}. Its
analysis, which proves Theorem~\ref{thm dyn ridgereg inf},
is given in Appendix~\ref{sec:ridge_proof}.

\section{Sampling from a Low-Rank Approximation}\label{sec rec}

Our algorithm for low-rank approximation is $\textsc{BuildLowRankFactors}$,
Algorithm~\ref{alg:dyn lowrank}, given below. As discussed in the introduction,
it uses $\textsc{LenSqSample}$, Algorithm~\ref{alg:lensqsample},
to reduce to a matrix whose size is independent of the input size, beyond
log factors, as well as Projection-Cost Preserving sketches, QR factorization, and
leverage-score sampling. Its analysis, proving Theorem~\ref{thm lra inf},
is given in Appendix~\ref{sec:LRA_proof}.

\def\rkt{t}

\begin{algorithm}[htb]
 {\small
\caption{$\textsc{BuildLowRankFactors}
\\(\DynSampler, k, \hsig_\rA, \hsig_{k}, \eps, \tau)$}
\label{alg:dyn lowrank}
{\bf Input:}
$\DynSampler = \DynSamp(A)$ (Def.~\ref{def dynsamp}) for $A\in\R^{\rows\times\colms}$,
$k$ target rank,
$\hsig_\rA\le 1/\norm{A^+}$,
$\hsig_{k} \le \sigma_{k}(A)$,
$\eps$ an error parameter,
$\tau$ estimate of $\norm{A-A_{k}}_F^2$,
where $A_{k}$ is the best rank-$k$ approximation
to $A$\\
{\bf Output:} Small matrix $W$ and sampling matrices
$S$, $R$ \\ % and 
so that $\rank(AR W SA)=k$ and \\
$\norm{ARWSA-A} \le (1+\eps)\norm{A-A_k}$\\
\begin{algorithmic}[1]
\STATE $\lambda\gets \tau/k$, $Z_\lambda \gets 1/\sqrt{\lambda + \hsig_\rA^2}$,  $Z_{k} \gets 1/\hsig_{k}$
\STATE Choose $m_R = m_S = O(\hat m_S Z_\lambda^2 \norm{A}_F^2)$,\\
where $\hat m_{S} = O(\eps^{-2}\log k)$
\STATE $S, R_1\gets \textsc{LenSqSample}(\DynSampler, \mathbf{null}, m_S, m_R) $
% \COMMENT{$m_R = m_S = O(\hat m_S Z_\lambda^2 \norm{A}_F^2)$, $\hat m_{S} = O(\eps^{-2}\log k)$}

% \STATE $S \gets \textsc{LenSqEmbed}(A, \DynSampler, Z_\lambda, \texttt{row}, k, 0,$ \\ $ k\hat m_{S})$
%  $\qquad$  \COMMENT{Alg.~\ref{alg:lensqembed}; 
% here $\hat m_{S} = O(\eps^{-2}\log k)$, $S$ has $m_{S} = O(\hat m_{S} Z_\lambda^2 \norm{A}_F^2)$ rows}
% \STATE $R_1 \gets \textsc{LenSqEmbed}(\SS A, \DynSamp(A), Z_\lambda, \texttt{col}, k, 0, $ \\ $k\hat m_{\SS })$;
\STATE Apply Alg.~1 and Thm.~1 of \cite{cohen2017input} to $\SS AR_1$, get col. sampler $R_2$
\COMMENT{ $m_{R_2} = O(\eps^{-2} k\log k)$}
\STATE Apply Alg.~1 and Thm.~1 of \cite{cohen2017input} to $\SS AR_1R_2$, get row sampler $S_2$
\COMMENT{ $m_{S_2} = O(\eps^{-2} k\log k)$}
\STATE $V \gets$ top-$k$ right singular matrix of $S_2\SS AR_1R_2$
\STATE $U,\underline{\ \ } \gets \textsc{QR}(\SS  A R_1R_2 V)$
\\ \COMMENT{$U$ has orthonormal cols, $\SS  A R_1R_2 V= UC$ for matrix $C$}
% \STATE $R_3 \gets \textsc{LenSqEmbed}(\SS A, \DynSamp(A), Z_{k}, \texttt{col}, k, 0, $ $\eps^{-1} k\hat m_{R_3})$ 
\STATE Choose $m_{R_3} = O(\hat m_{R_3}\eps^{-1} Z_{k}^2\norm{A}_F^2)$,
\\ where $\hat m_{R_3} = O(\eps_0^{-2}\log k + \eps^{-1})$, $\eps_0$ a small constant
\STATE $R_3 \gets \textsc{LenSqSample}(\DynSampler, S, m_S, m_{R_3})$ 
% $\quad$ \COMMENT{$m_{R_3} = O(\hat m_{R_3}\eps^{-1} Z_{k}^2\norm{A}_F^2)$ where $\hat m_{R_3} = O(\eps_0^{-2}\log k + \eps^{-1})$}
\STATE Let $f(k,C)$ be the function returning the value \\ $\qquad m_{R_4} = O(\eps_0^{-2}k\log k + \eps^{-1}k)$
\STATE $R_4^\top,\underline{\ \ \ } \gets \textsc{LevSample}((U^\top {\SS }AR_3)^\top, \log(m_{R_3}), f())$ \\
\COMMENT{Alg.~\ref{prelim_alg:levsample}}
% \nadiia{this algorithm is using the leverage score sampler as a subroutine.}
% \klcin{now there's a stub for the lev score sampler.}
\STATE $R\gets R_3 R_4$
\STATE $W \gets (U^\top {\SS }AR)^+U^\top$
\STATE return $W$, ${\SS }$, $R$
\end{algorithmic}
}% small
\end{algorithm}

\section{Experiments}
\label{sec:exp}
We evaluate the empirical performance of our algorithm on both synthetic and real-world datasets. All of our experiments were done in Python and conducted on a laptop with a 1.90GHz CPU and 16GB RAM. Prior work~\cite{quantum_practice} suggests the tree data structure is only faster than the built-in sampling function when the matrix size $\max\{n, d\}$ is larger than $10^6$. Hence we follow the implementation in~\cite{quantum_practice} that directly uses the built-in function. For a fair comparison, we also modified the code in~\cite{quantum_practice}, which reduces the time to maintain the data structure by roughly 30x. For each experiment, we took an average over $10$ independent trials. 

We note that we do not compare with classical sketching algorithms for several reasons. First, there is no classical contender with the same functionality as ours. This is because our dynamic algorithms support operations not seen elsewhere: sublinear work for regression and low-rank approximation, using simple fast data structures that allow, as special cases, row-wise or column-wise updates. Second, unlike dynamic algorithms where a sketch is maintained, our algorithms are not vulnerable to updates based on prior outputs, whether adversarially, or due to use in the inner loop of an optimization problem. 
This is because our algorithms are based on independent sampling from the exact input matrix. 

\subsection{Low-Rank Approximation}
\label{exp:lra}

We conduct experiments on the following datasets:
\vspace{-0.1in}
\begin{itemize}
    \setlength{\itemsep}{0ex}  
    \setlength{\parsep}{0ex}  
    \setlength{\parskip}{0ex}  
    \item \textbf{KOS data.}\footnote{The \href{https://archive.ics.uci.edu/ml/datasets/Bag+of+Words}{Bag of Words Data Set} from the UCI Machine Learning Repository.} A word frequency dataset. The matrix represents word frequencies in blogs and has dimensions 3430 × 6906 with 353160 non-zero entries.
    
    \item \textbf{MovieLens 100K.}~\cite{movielen} A movie ratings dataset, which consists of a preference matrix with 100,000 ratings from 611 users across 9,724 movies.
\vspace{-0.1in}
\end{itemize}

We compare our algorithms with the implementations in~\cite{quantum_practice}, which are based on the algorithms in~\cite{fkv04} and ~\cite{tang2019quantum}. We refer to this algorithm as ADBL henceforth. 
%\textcolor{red}{N: defining ADBL here.} 
For the KOS dataset, we set the number of sampled rows and columns to be $(r,c) = (500, 700)$ for both algorithms. For the MovieLens dataset we set $(r,c) = (300, 500)$. We define the error $\eps$ = $\|A - Y\|_F / \|A - A_k\|_F - 1$, where $Y$ is the algorithm's output and $A_k$ is the best $k$-rank approximation. Since the regime of interest is $k \ll n$, we vary $k$ among $\{10, 15, 20\}$.

The results are shown in Table~\ref{tab:movie}. %and~\ref{tab:KOS}. 
We first report the total runtime, which includes the time to maintain the data structure and then compute the low-rank approximation. We also report the query time, which excludes the time to maintain the data structure.
%process the sampling probabilities. 
From the table we see that both algorithms can achieve $\eps \approx  0.05$ in all cases. The query time of ours is about 6x-faster than the ADBL algorithm in~\cite{quantum_practice}, and even for the total time, our algorithm is much faster than the SVD. Although the accuracy of ours is slightly worse, in Appendix~\ref{sec:appendix_exp_lra} we show by increasing the sample size slightly, our algorithm achieves the same accuracy as ADBL~\cite{quantum_practice}, but still has a faster runtime. 

We remark that the reason our algorithm only needs half of the time to compute the sampling probabilities is that we only need to sample rows or columns according to their squared length, but the algorithm in~\cite{quantum_practice} also needs to sample entries for each sampled row according to the squared values of the entries.

\begin{table}[t]
\caption{Performance of our algorithm and ADBL on MovieLen 100K and KOS data, respectively.}%. The above is for MovieLen 100K and the below is for KOS dats}
\label{tab:movie}
\centering
\begin{tabular}{|c|c|c|c|}
\hline
& $k=10$ & $k=15$ & $k=20$ \\
\hline
$\eps$(Ours) & 0.0416 & 0.0557 & 0.0653\\
\hline
$\eps$(ADBL) & 0.0262 & 0.0424 & 0.0538\\
\hline
\hline
 Runtime  & \multirow{2}{*}{0.125s} & \multirow{2}{*}{0.131s} & \multirow{2}{*}{0.135s} \\
(Ours, Query) & & & \\
\hline
Runtime & \multirow{2}{*}{0.181s} & \multirow{2}{*}{0.183s} & \multirow{2}{*}{0.184s} \\
(Ours, Total) & & & \\
\hline
Runtime& \multirow{2}{*}{0.867s} & \multirow{2}{*}{0.913s} & \multirow{2}{*}{1.024s} \\
(ADBL, Query) & & & \\
\hline
Runtime & \multirow{2}{*}{0.968s} & \multirow{2}{*}{1.003s} & \multirow{2}{*}{1.099s} \\
(ADBL, Total)& & & \\
\hline
Runtime of SVD & \multicolumn{3}{c|}{2.500s}\\
\hline
\end{tabular}
\begin{tabular}{|c|c|c|c|}
\hline
& $k=10$ & $k=15$ & $k=20$ \\
\hline
$\eps$(Ours) & 0.0397 & 0.0478 & 0.0581\\
\hline
$\eps$(ADBL)  & 0.0186 & 0.0295 & 0.0350\\
\hline
\hline
 Runtime  & \multirow{2}{*}{0.292s} & \multirow{2}{*}{0.296s} & \multirow{2}{*}{0.295s} \\
(Ours, Query) & & & \\
\hline
Runtime & \multirow{2}{*}{0.452s} & \multirow{2}{*}{0.455s} & \multirow{2}{*}{0.452s} \\
(Ours, Total) & & & \\
\hline
Runtime& \multirow{2}{*}{1.501s} & \multirow{2}{*}{1.643s} & \multirow{2}{*}{1.580s} \\
(ADBL, Query) & & & \\
\hline
Runtime & \multirow{2}{*}{1.814s} & \multirow{2}{*}{1.958s} & \multirow{2}{*}{1.897s} \\
(ADBL, Total)& & & \\
\hline
Runtime of SVD & \multicolumn{3}{c|}{36.738s}\\
\hline
\end{tabular}
\end{table}
% \vspace{-0.2in}

% \begin{table}[t]
% \caption{Performance of our algorithm and ADBL on KOS data.}\label{tab:KOS}
% \centering
% \begin{tabular}{|c|c|c|c|}
% \hline
% & $k=10$ & $k=15$ & $k=20$ \\
% \hline
% $\eps$(Ours) & 0.0397 & 0.0478 & 0.0581\\
% \hline
% $\eps$(ADBL)  & 0.0186 & 0.0295 & 0.0350\\
% \hline
% \hline
%  Runtime  & \multirow{2}{*}{0.292s} & \multirow{2}{*}{0.296s} & \multirow{2}{*}{0.295s} \\
% (Ours, Query) & & & \\
% \hline
% Runtime & \multirow{2}{*}{0.447s} & \multirow{2}{*}{0.460s} & \multirow{2}{*}{0.453s} \\
% (Ours, Total) & & & \\
% \hline
% Runtime& \multirow{2}{*}{1.302s} & \multirow{2}{*}{1.437s} & \multirow{2}{*}{1.498s} \\
% (ADBL, Query) & & & \\
% \hline
% Runtime & \multirow{2}{*}{55.450s} & \multirow{2}{*}{56.175s} & \multirow{2}{*}{56.340s} \\
% (ADBL, Total)& & & \\
% \hline
% Runtime of SVD & \multicolumn{3}{c|}{36.738s}\\
% \hline
% \end{tabular}
% \end{table}

% \subsection{Ridge Regression}

% We also consider the ridge regression problem: 
%\[
% $X^* := 
% \min_{X \in \R^{d \times d'} } \norm{ A X - B }_F^2 + \lambda \norm{X}_F^2$,
%\]
% where $A \in \R^{n \times d}$, $B \in \R^{n \times d'}$. 
% Due to space constraints, the details are given in Appendix~\ref{sec:appendix_exp_ridge}.

\subsection{Ridge Regression}
\label{sec:exp_ridge}
In this section, we consider the problem 
\[
X^* := 
\min_{X \in \R^{d \times d'} } \norm{ A X - B }_F^2 + \lambda \norm{X}_F^2,
\]
where $A \in \R^{n \times d}$, $B \in \R^{n \times d'}$. We do experiments on the following dataset with $\lambda = 1$:
\begin{itemize}
\setlength{\itemsep}{0ex}  
    \setlength{\parsep}{0ex}  
    \setlength{\parskip}{0ex}  
    \item \textbf{Synthetic data.} We generate the rank-$k$ matrix $A$ as~\cite{quantum_practice} do. Particularly, suppose the SVD of $A$ is $A = U\Sigma V^{\top}$. We first sample an $n \times k$ Gaussian matrix, then we perform a QR-decomposition $G = QR$, where $Q$ is an $n \times k$ orthogonal matrix. We then simply set $U = Q$ and then use a similar way to generate $V$. We set $A \in \R^{7000 \times 9000}$, $B \in \R^{7000 \times 1}$.
    \item \textbf{YearPrediction.}\footnote{ \href{https://archive.ics.uci.edu/ml/datasets/yearpredictionmsd}{YearPredictionMSD Data Set}} A dataset that collects 515345 songs and each song has 90 attributes. The task here is to predict the release year of the song. $A \in \R^{515345 \times 90}$, $B \in \R^{515345 \times 1}$.
    \item \textbf{PEMS data.} \footnote{\href{https://archive.ics.uci.edu/ml/datasets/PEMS-SF}{PEMS-SF Data Set}} The data describes the occupancy rate of different car lanes of San Francisco bay area freeways. Each row is the time series for a single day. The task on this dataset is to classify each observed day as the correct day of the week, from Monday to Sunday. $A \in \R^{440 \times 138672}$, $B \in \R^{440 \times 1}$.
\end{itemize}

We define the error $\eps = \|X - X^*\|_F / \|X^*\|_F$, given the algorithm output $X$. For synthetic data, we set the number of sampled rows and columns to be $r$ and $c$. For the YearPrediction data, the number of columns is small, and hence we only do row sampling, and likewise, for the PEMS data, we only do column sampling. We did not find an implementation for the ridge regression problem in the previous related work. Therefore, here we list the time to compute the closed-form optimal solution $X^* = (A^{\top} A + \lambda I)^{-1} A^{\top} B$ or $X^* = A^{\top}(A A^{\top} + \lambda I)^{-1} B$, as a reference.

The results are shown in Table~\ref{tab:syn} and~\ref{tab:year}. From the tables we can see that for  synthetic data, the algorithm can achieve an error $\eps < 0.1$ when only sampling less than $10\%$ of the rows and columns. Also, the total runtime is about 40x-faster than computing the exact solution. For the YearPrediction and PEMS data, the bottleneck of the algorithm becomes the time to compute the sample probabilities, but the query time is still very fast and we can achieve an error $\eps < 0.1$ when only sampling a small fraction of the rows or columns.

\begin{table}[t]
\caption{Performance of our algorithm on synthetic data.}\label{tab:syn}
\centering
\begin{tabular}{|c|c|c|c|}
\hline
$(r, c)$& $300, 500$ & $500, 800$ & $1000, 1500$ \\
\hline
$\eps$(Ours) & 0.1392 &  0.0953& 0.0792\\
\hline
\hline
 Runtime  & \multirow{2}{*}{0.021s} & \multirow{2}{*}{0.042s} & \multirow{2}{*}{0.148s} \\
(Query) & & & \\
\hline
Runtime & \multirow{2}{*}{0.557s} & \multirow{2}{*}{0.568s} & \multirow{2}{*}{0.667s} \\
(Total) & & & \\
\hline
Exact $X*$ & \multicolumn{3}{c|}{24.074s}\\
\hline
\end{tabular}
\end{table}

\begin{table}[h]
\caption{Performance of our algorithm on YearPrediction data and PEMS data, respectively.}\label{tab:year}
\centering
\begin{tabular}{|c|c|c|c|}
\hline
$r = $& $1000$ & $3000$ & $5000$ \\
\hline
$\eps$(Ours) & 0.1070 & 0.0633 & 0.0447\\
\hline
\hline
 Runtime  & \multirow{2}{*}{0.031s} & \multirow{2}{*}{0.037s} & \multirow{2}{*}{0.059s} \\
(Query) & & & \\
\hline
Runtime & \multirow{2}{*}{0.213s} & \multirow{2}{*}{0.229s} & \multirow{2}{*}{0.245s} \\
(Total) & & & \\
\hline
Exact $X^*$ & \multicolumn{3}{c|}{0.251s}\\
\hline
\end{tabular}
\begin{tabular}{|c|c|c|c|}
\hline
$c = $ & $15000$ & $25000$ & $35000$ \\
\hline
$\eps$(Ours) & 0.1778 & 0.1397 & 0.1130\\
\hline
\hline
 Runtime  & \multirow{2}{*}{0.234s} & \multirow{2}{*}{0.381s} & \multirow{2}{*}{0.532s} \\
(Query) & & & \\
\hline
Runtime & \multirow{2}{*}{0.473s} & \multirow{2}{*}{0.628s} & \multirow{2}{*}{0.777s} \\
(Total) & & & \\
\hline
Exact $X*$ & \multicolumn{3}{c|}{0.972s}\\
\hline
\end{tabular}
\end{table}

\section*{Acknowledgements.} Honghao Lin and David Woodruff would like to thank for partial support from the National Science Foundation (NSF) under Grant No. CCF-1815840.

\newpage
\bibliographystyle{alpha}
\bibliography{p.bib}

\appendix
\section{Preliminaries}
\label{sec:appendix_prelim}

In, this section, we provide proofs for our theorems in Section \ref{subsec nota}. We first provide the description of the leverage score sampling data structure (Algorithm~\ref{prelim_alg:levsample}). We note that the advantage of the current version compared to the standard leverage score sampling (see, e.g., Section~2.4 in the survey in~\cite{woodruff2014sketching}) is that it saves an $O(\log n)$ factor.
We need the following data structure.

\begin{definition}[Sampling Data Structure]
\label{def samp}
Given a matrix $A\in\R^{\rows\times\colms}$, a column selection matrix $\Lambda$ such that $k=\rank(A\Lambda) = \rank(A)$, and $\lambda_s\ge 1$, the data structure 
$\textsc{Samp}(A, \Lambda, \lambda_s)$ consists of the following: 
\begin{itemize}
 \item $SA$, where $S\in\R^{m_S \times \rows}$ is a sketching matrix as in Lemma~\ref{lem embed}, with $m_S = O(\rA\log(\rA)/\eps_0^2)$, chosen to be an $\eps_0$-embedding with failure probability $1/100$, for fixed $\eps_0$;
  \item $C$, where $[Q, C] \gets \textsc{QR}(SA\Lambda)$, the QR decomposition of $SA\Lambda$, i.e.,  $SA\Lambda = QC$, $Q$ has orthonormal columns and $C$ is triangular and invertible (since $A\Lambda$ has full column rank);
 \item $C_0$, where $[Q_0, C_0] \gets \textsc{QR}(SA)$;
 \item The data structure of Lemma~\ref{lem weighted tree}, built to enable sampling $i\in [\rows]$ with probability $p_i\gets \norm{Z_{i,*}}^2/\norm{Z}_F^2$ in $O(\log\rows)$ time, where $Z \gets A\Lambda (C^{-1}G)$, with $G\in\R^{k\times m_G}$ having independent ${\cal N}(0,1/m_G)$ entries, and $m_G = \Theta(\lambda_s)$.
\end{itemize}
\end{definition}

% \textcolor{red}{N:Adding def of a sampling data structure.}

\begin{algorithm}[hbt!]
 {\small
\caption{$\textsc{MatVecSampler}(A, \Sampler, W, \samplesize, \nu)$}
\label{alg:matvecsampler}
{\bf Input:} $A \in \R^{\rows\times\colms}$, data structure $ \Sampler$ (Def.~\ref{def samp}), $W\in\R^{\colms\times m_W}$,
desired number of samples $\samplesize$, normalizer $\nu$, where $\nu= \frac{1}{6 \rA \rows^{1/\lambda_s}}$ by default if unspecified\\
{\bf Output:} $L\in\R^{v\times\colms}$, encoding $v$ draws from $i\in [\rows]$ chosen with approx. probability $q_i \defeq \norm{A_{i,*}W}^2/\norm{AW}_F^2$\\
\begin{algorithmic}[1]
\STATE $N \gets \norm{C_0W}_F^2$, where $C_0$ is from $\Sampler$ 
\STATE if $N == 0$:
\STATE $\qquad$ return $\textsc{Uniform}(v,\rows)$ \label{alg uni} \COMMENT{Alternatively, raise an exception here}
\STATE $L\gets 0_{v\times \rows}$, $z\gets 0$ \label{alg mvs C}
\STATE while $z<\samplesize$:
\STATE $\qquad$ Choose $i\in [\rows]$ with probability $p_i$ using $\Sampler$ \label{alg crude}
\STATE $\qquad$ $\tq_i \gets \norm{A_{i, *}W}^2/N$
\STATE $\qquad$ With probability $\nu\frac{\tq_i}{p_i}$, accept $i$: set $L_{z,i} = 1/\sqrt{v \tq_i}$; $z\gets z+1$\label{alg matvecsample reject}
\STATE return $L$
\end{algorithmic}
}% small
\end{algorithm}

% \textcolor{red}{N:Adding def of a mat-vec sampler, and corresponding lemmas for running time and correctness.}

%DW Here

\begin{lemma}[Sampling Data structure]\label{lem Samp time}
 The data structure $\Samp(A, \Lambda,  \lambda_s)$, from Definition \ref{def samp}, can be constructed in $O(\lambda_s (\nnz(A) + \rA^2) + \colms^\omega)$ time.
\end{lemma}

\begin{proof}
 The time needed to compute $SA$ is $O( \nnz(A) + k^{\omega}\textrm{poly}(\log\log(k))+k^{2+o(1)}\eps_0^{-2} )$.
Computing the $QR$ factorization of $SA$ takes $O(\colms^\omega)$ time, by first
 computing $(SA)^\top(SA)$ for the $m_S\times d$ matrix $SA$,
 and then its Cholesky composition,
 using ``fast matrix'' methods for both, and using $m_S\le \colms$.
 This dominates the time for the similar factorization of $SA\Lambda$.
 
 The $Z$ matrix can be computed in $O(\lambda_s(\nnz(A) + \rA^2))$ time, by appropriate order of multiplication, and this dominates
 the time needed for building the data structure of Lemma~\ref{lem weighted tree}. Adding these terms, and using $m\le\nnz(A)$, the result follows.
\end{proof}

\begin{lemma}[MatVecSampler Analysis]\label{lem matvecsampler}
 Given constant $c_0>1$ and small enough constant $\eps_0>0$,
 and $\Sampler$ for $A$,
 there is an event $\cE$ holding with failure probability at most $1/\rA^{c_0}$, so that if $\cE$ holds,
 then when called with $\nu\gets \frac{1}{6 \rA \rows^{1/\lambda_s}}$,
 the probability is $(1\pm\eps_0) q_i$  that the accepted index in
 Step~\ref{alg matvecsample reject} of $\textsc{MatVecSampler}$ is $i\in [\rows]$,
 where $q_i\defeq \norm{A_{i,*}W}^2/\norm{AW}_F^2$.
 The time taken is $O(m_W\colms (\colms + \samplesize k \rows^{1/\lambda_s})$, where $k=\rank(A)$.
\end{lemma}

\begin{proof}
We need to verify that the quantity in question is a probability, that is, that $\nu\frac{\tq_i}{p_i}\in (0,1)$,
when $\nu=\frac{1}{6 k \rows^{1/\lambda_s}}$.

From Lemma~\ref{lem embed},
if $m_S = O(\eps_0^{-2}\rA)$ for $\eps_0>0$, then with failure probability $1/k^{c_0+1}$,
$S$ will be a subspace $\eps_0$-embedding for $\colspan(A)$, that is, for $A\Lambda$ and
for $A$, using $\rank(A)=k$. The event $\cE$ includes the condition that $S$ is indeed an $\eps_0$-embedding.
If this holds for $S$, then from standard arguments,
$A\Lambda C^{-1}$ has singular values all in $1\pm\eps_0$, and $\norm{A_{i,*}\Lambda C^{-1}}^2 = (1\pm O(\eps_0))\tau_i$,
where again $\tau_i$ is the $i$'th leverage score.

(We have $\norm{A\Lambda C^{-1}x}= (1\pm\eps_0) \norm{SA\Lambda C^{-1}x} = (1\pm\eps_0)\norm{x}$,
for all $x$, since $SA=QC$.) This implies that  for $Z,G$ in the construction of $\Samp(A, \Lambda, \lambda_s)$,
\[
\norm{Z}_F^2 = \norm{A\Lambda C^{-1}G}_F^2 \le (1+\eps_0)\norm{G}_F^2.
\]
Since $m_G \norm{G}_F^2$ is $\chi^2$ with $km_G$ degrees of freedom, with failure probability at most 
$\exp(-\sqrt{k\lambda_s}/2)$ (using $m_G=\Theta(\lambda_s)$,
it is at most $3km_G$ (\cite{laurent2000}, Lemma~1),
so $\norm{G}_F^2\le 3k$ with that probability. Our event $\cE$ also includes
the condition that this bound holds. Thus under this condition, $\norm{Z}_F^2 \le 3(1+\eps_0)k$.

From Lemma~\ref{lem JL} and the above characterization of $\tau_i$, for the $Z$ of $\Samp(A,\Lambda, \lambda_s)$,
$\norm{Z_{i,*}}^2 = \norm{A_{i,*}\Lambda C^{-1}G}^2 \ge (1-O(\eps_0))\tau_i/\rows^{1/\lambda_s}$.

Putting these together, we have
\begin{equation}\label{eq p bound}
p_i = \frac{\norm{Z_{i,*}}_F^2}{\norm{Z}^2} \ge (1-O(\eps_0)) \frac{ \tau_i/\rows^{1/\lambda_s}}{3k}.
\end{equation}

Using the $\eps_0$-embedding property of $S$, 
\begin{equation}\label{eq CA}
\norm{C_0 W}_F^2 = \norm{Q_0C_0W}_F^2 =  \norm{SAW}_F^2 = (1\pm2\eps_0)\norm{AW}_F^2,
\end{equation}
and so, letting $A= UC_1$ for $U$ with orthonormal columns, we have, for small enough $\eps_0$,
\begin{align*}
(1-2\eps_0) \tq_i  & \le \frac{ \norm{A_{i, *}W}^2}{\norm{AW}_F^2} 
	 =  \frac{ \norm{U_{i, *}C_1 W}^2}{\norm{UC_1W}_F^2}  =  \frac{ \norm{U_{i, *}C_1 W}^2}{\norm{C_1W}_F^2}
	 \le \frac{ \norm{U_{i, *}}^2 \norm{C_1 W}^2}{\norm{C_1W}_F^2} \le \tau_i.
\end{align*}
Putting this bound with \eqref{eq p bound} we have
\[
\frac{\tq_i}{p_i}
	\le \frac{\tau_i/(1-2\eps_0)}{(1-O(\eps_0)) \tau_i/\rows^{1/\lambda_s}3(1+\eps_0)k} \le 3(1+O(\eps_0))k\rows^{1/\lambda_s}.
\]
%\begin{align*}
%\frac{\tq_i}{p_i}
%	   = \frac{ \norm{A_{i, *}W}^2}{\norm{CW}_F^2}\frac{\norm{Z}_F^2}{\norm{Z_{i,*}}^2}
%	 \le \frac{ \norm{A_{i, *}C^{-1}}^2 \norm{CW}_F^2}{\norm{CW}_F^2}\frac{3(1+\eps_0)n}{\norm{A_{i,*}C^{-1}}^2/\rows^{1/\lambda_s}} 
%	 \le  3(1+\eps_0)n \rows^{1/\lambda_s},
% \end{align*}
 so that $\nu \frac{\tq_i}{p_i} = \frac{1}{6 k \rows^{1/\lambda_s}}\frac{\tq_i}{p_i}\le (1+O(\eps_0))/2 \le 1$ for small enough $\eps_0$.
 Using \eqref{eq CA} we have $\tq_i = (1\pm 2\eps_0) q_i$.
  Thus the correctness condition of the lemma follows, for small enough $\eps_0$.
 
 Turning to time: the time to compute $C_0 W$ is $O(\colms^2 m_W)$.
 Each iteration takes $O(\log \rows + \colms m_W)$, for
 choosing $i$ and computing $\tq_i$, and these steps dominate the time. As usual for rejection sampling,
 the expected number of iterations is $O(\samplesize k \rows^{1/\lambda_s})$. Adding these expressions
 yields the expected time bound, folding a factor of $\log \rows$ in by adjusting $\lambda_s$ slightly.
\end{proof}

\begin{algorithm}[H]
 {\small
\caption{$\textsc{LevSample}(A, \mu_s, f())$}
\label{prelim_alg:levsample}
{\bf Input:} $A \in \R^{\rows\times\colms}$,  $\mu_s\ge 1$ specifying runtime tradeoff,
function $f(\cdot) \rightarrow \mathbb{Z}_+$ returns a target sample size (may be just a constant) \\
{\bf Output:} Leverage score sketching matrix $L$, column selector $\Lambda$
}% small

\begin{algorithmic}[1]
\STATE Run an algorithm to compute $k=\rank(A)$  and obtain  $\Lambda\in\R^{\colms\times k}$, a subset of $k$ lin. indep. columns of $A$  \COMMENT{for example Theorem 1.5  in~\cite{chepurko2022near}}\label{alg reduceW}
\STATE Construct $\Sampler \gets \Samp(A\Lambda, I, \lambda_s)$, use $C$ from it; \COMMENT{Definition \ref{def samp}} \label{alg makesamp}
\STATE $W\gets C^{-1}G'$, where $G'\in\R^{\rA\times m_{G'}}$ with ind. ${\cal N}(0,1/m_{G'})$ entries \COMMENT{$m_{G'} = \Theta(\log\rows)$} \label{alg makeW}
\STATE $L\gets \textsc{MatVecSampler}(A\Lambda , \Sampler, W, f(k,C), \nu = 1/6\rows^{1/\lambda_s})$\\ $\qquad$ \COMMENT{Algorithm \ref{alg:matvecsampler}, sample size $f(k,C)$, normalizer $\nu$}\label{alg step mvs}
\STATE return $L, \Lambda$ 
\end{algorithmic}
\end{algorithm}

% \textcolor{red}{N: I added back the description of the algorithm, somehow this disappeared. One thing to note is that this algorithm and the corresponding theorem require an algorithm to compute a rank $k$ basis, which no longer appears in this paper. We can just cite this subroutine from the other paper. }
\begin{theorem}[Leverage Score Data Structure, Theorem \ref{prelim:thm levsample} restated]
\label{thm levsample}
Let $k=\rank(A)$, and choose $\mu_s \ge 1$.
 Algorithm~\ref{prelim_alg:levsample} ($\textsc{LevSample}(A,\mu_s, f(\cdot))$)
 uses space $O(\rows + \rA^\omega\log\log(\rows\colms))$,
 not counting the space to store $A$, and runs in time
\[
O(\mu_s\nnz(A) + \rA^\omega\textrm{poly}(\log\log(k)) + \rA^{2+o(1)} +  \samplesize k \rows^{1/\mu_s}),
\]
where $\samplesize$ is the sample size.
For $\samplesize=\eps^{-2}k\log k$, this bound can be expressed as
$O(\mu_e\nnz(A) + \rA^\omega\textrm{poly}(\log\log(k)) + \eps^{-2-1/\mu_e}k^2)$ time, for $\mu_e\ge 1$.
\end{theorem}

Note: we can get a slightly smaller running time by including more rounds of rejection sampling: the first
round of sampling needs an estimate with failure probability totaled over for all $\rows$ rows, while another round would only need such
a bound for $\samplesize \rows^{1/\lambda_s}$ rows; this would make the bound $\samplesize^{1+1/\lambda_s} k \rows^{1/\lambda_s^2}$,
which would be smaller when $\samplesize\ll \rows$. However, in the latter case, the term $ \samplesize k \rows^{1/\lambda_s}$ is dominated
by the other terms anyway, for relevant values of the parameters. For example if $\samplesize\le \rows$ and $\samplesize k\le\nnz(A)$
does not hold, then sampling is not likely to be helpful. Iterating $\log\log \rows$ times, a bound with leading term $O(\nnz(A)(\log\log \rows + \log \samplesize))$
is possible, but does not seem interesting.
%DW: :-)

\begin{proof}
 Step~\ref{alg makesamp}, building $\Samp(A\Lambda , \lambda_s)$,
 take $O(\lambda_s (\nnz(A) + \rA^2) + \rA^\omega)$ time, with $\colms$ in Lemma~\ref{lem Samp time} equal to $\rA$ here.
 
 From Lemma~\ref{lem matvecsampler}, the running time of $\textsc{MatVecSampler}$ is
 $O(k^2\log \rows + \samplesize k^2 (\log\rows) \rows^{1/\lambda_s})$, mapping $\colms$ of the lemma to $k$, $m_W$ to $m_{G'}=O(\log\rows)$.
 However, since the normalizer $\nu$ is a factor of $k$ smaller than assumed in Lemma~\ref{lem matvecsampler},
 the runtime in sampling is better by that factor. Also, we subsume the second $\log \rows$ factor by adjusting $\lambda_s$.
 
 The cost of computing $C^{-1}G'$ is $O(k^2\log\rows)$; we have a runtime of
 \begin{align*}
 O(\nnz(A) & + \rA^\omega\textrm{poly}(\log\log(k)) + k^{2+o(1)}) + O(\lambda_s (\nnz(A) + \rA^2) + \rA^\omega) + O(k^2\log \rows + \samplesize k \rows^{1/\lambda_s})
 	\\ & = O(\lambda_s\nnz(A) + \rA^\omega\textrm{poly}(\log\log(k)) + \rA^{2+o(1)} +  \samplesize k \rows^{1/\lambda_s}),
 \end{align*}
 as claimed.
 
 Finally, suppose $\samplesize=\eps^{-2}k \log k$, as suffices for an $\eps$-embedding.
 If $\samplesize k\rows^{1/\lambda_s} \le \nnz(A) + k^\omega$, then the bound follows. Suppose not.
 If $\rows \ge k^\omega$, then
 \[
 \eps^{-2} \ge \rows^{1-1/\lambda_s}/k^2\log(k) \ge \rows^{1-1/\lambda_s - 2/\omega}/\log(\rows)
 \]
 and so $\eps^{-2} \ge \rows^{\gamma}$, for constant $\gamma>0$,
 implying $\eps^{-2/\lambda_s\gamma'}\ge \rows^{1/\lambda_s}\log \rows$, for constant $\gamma' < \gamma$.
 When $k^\omega\ge \rows$,
 \[
 \eps^{-2} \ge k^{\omega-2- 1/\omega\lambda_s}/\log(k) \ge k^\gamma,
 \]
 for a constant $\gamma>0$, so that $\eps^{-\omega/\lambda_s\gamma'} \ge \rows^{1/\lambda_s}\log k$,
 for a constant $\gamma'<\gamma$.
 Using $\lambda_e$, a constant multiple of $\lambda_s$, to account for constants, the result follows. 
\end{proof}

We can also use \emph{length-squared sampling} to obtain subspace embeddings. In Section~\ref{subsec dyn lev}
we have given a data structure and algorithm that implements length-squared sampling. To analyze length-squared sampling in the context of ridge regression, we show the following structural observations about ridge regression.

\begin{lemma}[Block SVD]\label{lem Al}
 Let $A$ be such that $\rA=\rank(A)$, and suppose $A$ has thin SVD $A=U\Sigma V^\top$, implying $\Sigma\in\R^{\rA\times\rA}$.
 For $\lambda>0$, let $\Al = \twomat{A}{\sqrt{\lambda} VV^\top}$. For $b\in\R^\rows$, let
 $\hb=\twomat{b}{0_\colms}$. Then for all  $x\in\colspan(V)$, the ridge regression loss
\[
\norm{Ax-b}^2 + \lambda\norm{x}^2 =  \norm{\Al x - \hb}^2,
\]
and ridge regression optimum
\[
	x^* = \argmin_{x\in\R^\colms} \norm{Ax-b}^2 + \lambda\norm{x}^2 = \argmin_{x\in\R^\colms} \norm{\Al x - \hb}^2.
	\]
The matrix $\Al$ has SVD $\Al = \twomat{U\Sigma D}{\sqrt{\lambda}VD}D^{-1}V^\top$, where
$D= (\Sigma^2 + \lambda\Iden_\rA)^{-1/2}$,
and $\norm{\Al^+}^2 = 1/(\lambda + 1/\norm{A^+}^2)$. We have
$ \norm{A_{i,*}}^2 \norm{\Al^+}^2\ge \norm{U_{i,*}\Sigma D}^2$ for $i\in [\rows]$.
\end{lemma}

\begin{proof}
 Since $x\in\colspan(V)$ has $x=Vz$ for some $z\in\R^\rA$, and since $V^\top V = \Iden_k$, it follows that
 $VV^\top x = Vz = x$, and so 
 \[
 \norm{\Al x - \hb}^2 = \norm{Ax-b}^2 + \norm{\sqrt{\lambda} VV^\top x - 0}^2 = \norm{Ax-b}^2 + \lambda\norm{x}^2 ,
 \]
 as claimed.
 
 The SVD of $\Al$ is
 $\Al = \twomat{U\Sigma D}{\sqrt{\lambda}VD} D^{-1} V^\top$, 
 where $D$ is defined as in the lemma statement, since the equality holds, 
 and both $\twomat{U\Sigma D}{\sqrt{\lambda}VD}$ and $V$ have
 orthonormal columns, and $D^{-1}$ has non-increasing nonnegative entries.
% (NB: this 
 Therefore $\Al^+ = VD\twomat{U\Sigma D}{\sqrt{\lambda}VD}^\top$. We have
  \begin{equation}\label{eq Alb}
  \Al^+\hb = VD^2\Sigma U^\top b = V\Sigma D^2 U^\top b,
  \end{equation}
  using that $\Sigma$ and $D$ are diagonal matrices.

By the well-known expression $x^*=A^\top(AA^\top + \lambda\Iden_\rows)^{-1} b$, and using the \emph{not}-thin SVD $A=\hU\hSigma\hV^\top$, with $\hSigma\in\R^{\rows\times\colms}$ and $\hU$ and $\hV$ orthogonal matrices,
\begin{equation}
\begin{split}
 x^*
 	   & = \hV\hSigma \hU^\top(\hU\hSigma\hSigma^\top\hU^\top + \lambda\hU\hU^\top)^{-1}b
	     \\
	    & = \hV\hSigma \hU^\top\hU(\hSigma\hSigma^\top + \lambda\Iden_\rows)^{-1}\hU^\top b\nonumber
	\\ 
	& = \hV\hSigma (\hSigma\hSigma^\top + \lambda\Iden_\rows)^{-1}\hU^\top b
	\\
	     &  = V\Sigma(\Sigma^2 + \lambda\Iden_\rA)^{-1} U^\top b\\
	     &
	      = V\Sigma D^2 U^\top b,\label{eq x*Sig}
\end{split}
\end{equation}
where the next-to-last step uses that $\hSigma$ is zero except for the top $\rA$ diagonal entries of $\hSigma$.
Comparing \eqref{eq Alb} and \eqref{eq x*Sig}, we have  $\Al^+\hb = x^*$.  Using the expression for $\Al^+$,
  $\norm{\Al^+}^2 = D^2_{1,1} = 1/(\lambda + 1/\norm{A^+}^2)$.
Finally, since $(\Al)_{i,*} = A_{i,*}$ for $i\in[\rows]$, and letting $\hat U = \twomat{U\Sigma D}{\sqrt{\lambda}VD}$,
\begin{equation*}
\begin{split}
\norm{A_{i,*}}^2  \norm{\Al^+}^2
	   & = \norm{(\Al)_{i,*}}^2  \norm{\Al^+}^2\\
	   & \ge \norm{(\Al)_{i,*}\Al^+}^2
	\\ 
	& = \norm{\hat U_{i,*} D^{-1} V^\top V D \hat U^\top}^2\\
	& = \norm{\hat U_{i,*}}^2 \\
	& = \norm{U_{i,*}\Sigma D}^2,
\end{split}
\end{equation*}
as claimed.
\end{proof}

\begin{definition}[Ridge Leverage-score Sample, Statistical Dimension]
 Let $A,\lambda,\Al$, and $D$ be as in Lemma \ref{lem Al}. Call ${\cal S} \subset [\rows]$
 a \emph{ridge leverage-score sample} of $A$ if each $i\in\cal S$ is chosen independently
 with probability at least $\norm{U_{i,*}\Sigma D}^2/\sd_\lambda(A)$, where the
 statistical dimension
 $\sd_\lambda(A) = \norm{U\Sigma D}_F^2 = \sum_{i\in [d]} \sigma_i^2/(\lambda+\sigma_i^2)$,
  recalling that $U\Sigma D$ comprises the top $n$ rows of the left singular matrix of $\Al$.
\end{definition}

\begin{lemma}[Length-squared sketch, Lemma \ref{prelim:lem Lsquare} restated]\label{lem Lsquare}
Given a matrix $A \in \mathbb{R}^{n \times d}$ and a sample size parameter $\samplesize \in [n]$,  let $m = O\left(\samplesize\norm{\Al^+}^2\norm{A}_F^2/\sd_\lambda(A)\right)$. Then, with probability at least $99/100$, 
 the set of $m$ length-squared samples contains a ridge leverage-score sample
 of $A$ of size $\samplesize$.
% 
% That is, there is $\samplesize=O(\eps_0^{-2}\rA\log\colms)$, where $\eps_0 > 0$,
% with corresponding $m_\hL = \eps_0^{-2} \norm{\Al^+}^2\norm{A}_F^2\log\colms$,
% so that with failure probability $1/\rA^c$, for given constant~$c$,
% $\hL$ is an $\eps_0$-embedding. If $\hat\Lambda$ is a matrix whose rows are a subset
% of the rows of the identity matrix, such that the rows of $\hat\Lambda A$ are those chosen
% by $\hL$, then $\rank(\hat\Lambda A) = \rank(A)$.
\end{lemma}
% \klcin{Is a more elaborate sampling scheme needed to ensure lev-score properties?}
% \nadiia{I think oversampling by a squared frobenius norm/(statistical dimension)*min regularlized singular value (up to $\lambda$) suffices for simulating leverage score sampling.}
Note that when $\lambda=0$, $\norm{\Al^+}=\norm{A^+}$, $\sd_0(A) = \rank(A)$,
and the ridge leverage-score samples are leverage-score samples.

\begin{proof}
We will show that $\hL$ contains within it a leverage-score sketching matrix; since
oversampling does no harm, this implies the result using the above lemma.
 
Using the thin SVD $A=U\Sigma V^\top$, and $A^+ = V\Sigma^+ U^\top$,
we have
\[
	\norm{A_{i,*}}\norm{A^+} 
		\ge \norm{A_{i,*}A^+}
		= \norm{U_{i,*} \Sigma V^\top V\Sigma^+ U^\top}
		= \norm{U_{i,*}},
\]
The expected number of times index $i\in [n]$ is chosen among $m_\hL$
length-squared samples, $p_i m_\hL$, is within a constant factor of $
\frac{\norm{A_{i,*}}^2}{\norm{A}_F^2} \samplesize \norm{\Al^+}^2\norm{A}_F^2/\sd_\lambda(A)
	\ge  \frac{\norm{U_{i,*}}^2}{\sd_\lambda(A)}\samplesize,$
using Lemma~\ref{lem Al},
an expectation at least as large as for a ridge leverage-score sample of size $\samplesize$.

\end{proof}

Finally, recall the Johnson-Lindenstraus Lemma, for sketching with a dense Gaussian matrix. 

\begin{lemma}[Johnson-Lindenstraus Lemma]\label{lem JL}
 For given $\eps>0$,
 if $P\subset \R^c$ is a set of $m\ge c$ vectors, and $G\in\R^{m\times c}$ has entries
 that are independent Gaussians with mean zero and variance $1/m$, then there is $m= O(\eps^{-2}\log(m/\delta))$ 
 such that with failure probability $\delta$,
 $\norm{Gx} = (1\pm\eps)\norm{x}$ for all $x\in P$. Moreover, there is $m_G=O(\mu)$
 so that $\norm{Gx}\ge \norm{x}/\rows^{1/\mu}$, with failure probability at most $1/\rows^2$.
\end{lemma}

\begin{definition}[Projection-Cost Preserving Sketch]
\label{def PCP}
Given a matrix $A \in \mathbb{R}^{n \times d}$, $\eps>0$ and an integer $k\in [d]$, a sketch $SA \in\mathbb{R}^{s \times d} $ is an $(\eps,k)$-column projection-cost preserving sketch of $A$ if for all rank-$k$ projection matrices $P$, $(1-\eps)\|A (I - P) \|_F^2\leq \| SA(I - P) \|^2_F \leq (1+\eps) \|A (I - P) \|_F^2.$
\end{definition}

There are several constructions of projection-cost preserving sketches known in the literature, starting with the work of Cohen et. al. ~\cite{cohen2015dimensionality, cohen2017input}. For our purposes, it suffices to use Theorem 1 from \cite{cohen2017input}.

% We will need an algorithm for leverage-score sampling. We provide the description of such an algorithm below, and a statement of its
% properties, but defer the proof to the supplementary material.
% %\klcin{citep our arxiv paper?}

% \begin{algorithm}[H]
%  {\small
% \caption{$\textsc{LevSample}(A, \mu_s, f())$}
% \label{alg:levsample}
% {\bf Input:} $A \in \R^{\rows\times\colms}$,  $\mu_s\ge 1$ specifying runtime tradeoff,
% function $f(\cdot) \rightarrow \mathbb{Z}_+$ returns a target sample size (may be just a constant) \\
% {\bf Output:} Leverage score sketching matrix $L$, column selector $\Lambda$
% }% small
% \end{algorithm}

% \begin{theorem}[Leverage Score Data Structure]\label{thm levsample}
% Let $k=\rank(A)$, and choose $\mu_s \ge 1$.
%  Algorithm~\ref{alg:levsample} ($\textsc{LevSample}(A,\mu_s, f(\cdot))$)
%  uses space $O(\rows + \rA^\omega\log\log(\rows\colms))$,
%  not counting the space to store $A$, and runs in time
% \[
% O(\mu_s\nnz(A) + \rA^\omega\log\log(\rows\colms) + \rA^2\log\rows +  \samplesize k \rows^{1/\mu_s}),
% \]
% where $\samplesize$ is the sample size.
% For $\samplesize=\eps^{-2}k\log k$, this bound can be expressed as
% $O(\mu_e\nnz(A) + \rA^\omega\log\log(n) + \eps^{-2-1/\mu_e}k^2)$ time, for $\mu_e\ge 1$.
% \end{theorem}
  
We can use the following lemma to translate from prediction error to solution error for regression problems. 

\begin{lemma}\label{lem p->x}
Let $\gamma_{A,b} = \frac{\norm{b}}{\norm{AA^+b}}$.
 Recall that $\kappa(A) = \norm{A} \norm{A^+}$.
 Suppose $\tx \in\colspan( A^\top)$, and for some $\eps_p\in (0,1)$, $\norm{A\tx - b} ^2 \le (1+\eps_p)\norm{\xi^*}^2$ holds, where $\xi^*= Ax^*-b$.
Then
 \begin{equation}\label{eq p->x}
 \norm{\tx-x^*} \le 2\sqrt{\eps_p} \norm{A^+} \norm{\xi^*} \le  2\sqrt{\eps_p} \sqrt{\gamma_{A,b}^2-1} \kappa(A) \norm{x^*}.
 \end{equation}
 This extends to multiple response regression using
 $\gamma_{A,B}^2 \defeq \frac{\norm{B}_F^2}{\norm{AA^+B}_F^2}$,
 by applying column by column to $B$, and extends to ridge regression, that is,
 $\Al$ with $\hB=\twomat{B}{0_{\colms\times d'}}$, as well.
\end{lemma}
 
Note that $x\in\colspan(A^\top)=\colspan(V)$ is no loss of generality, because
the projection $VV^\top x$ of $x$ onto $\colspan(A^\top)$ has $AVV^\top x = Ax$ and
$\norm{VV^\top x}\le \norm{x}$. So $\argmin_x \norm{Ax-b}^2+\lambda\norm{x}$
must be in $\colspan(A^\top)$ for $\lambda$ arbitrarily close to zero, and $A^+b\in\colspan(A^\top)$.

For the ridge problem $\min_x \norm{\Al x - \hb}$, we have
$\norm{\xi^*}^2 = \norm{\Al x^* - \hb} = \norm{Ax^* - b}^2 + \lambda\norm{x^*}^2$,
and recalling from Lemma~\ref{lem Al} that, when $A$ has SVD $A=U\Sigma V^\top$,
 $\Al$ has singular value matrix $D^{-1}$, where
$D= (\Sigma^2 + \lambda\Iden_\rA)^{-1/2}$, so that
$\kappa(\Al)^2 = (\lambda + \sigma_1^2)/(\lambda + \sigma_\rA^2)$,
where $A$ has singular values $\sigma_1,\ldots,\sigma_\rA$,
with $\rA=\rank(A)$.

\begin{proof}
 Since $x^* = A^+ b = A^\top (AA^\top)^+ b \in \colspan A^\top$, we have $\tx - x^* = A^\top z \in\colspan A^\top$,
 for some $z$.
 Since $A^+AA^\top = A^\top$, we have $\tx - x^* = A^\top z = A^+AA^\top z =   A^+A(\tx-x^*)$.
 From the normal equations for regression
 and the Pythagorean theorem,
 \[
 \norm{A(\tx - x^*)}^2 = \norm{A\tx - b}^2  - \norm{Ax^*-b}^2\le 4\eps_p\norm{\xi^*}^2,
 \]
 using $ \norm{A\tx - b} \le (1+\eps_p)\norm{\xi^*}$ and $\eps_p<1$. Therefore,
 using also submultiplicativity of the spectral norm,
 \begin{align}
 \norm{\tx - x^*}^2
 	   & = \norm{A^+A(\tx-x^*)}^2 \nonumber
	\\ & \le \norm{A^+}^2\norm{A(\tx - x^*)}^2 \nonumber
	\\ & \le  \norm{A^+}^2 4\eps_p\norm{\xi^*}^2, \label{eq dt*}
\end{align}
and the first inequality of \eqref{eq p->x} follows. For the second, we bound
\[
\frac{\norm{\xi*}^2}{\norm{x^*}^2}
	= \frac{\norm{b}^2 - \norm{AA^+b}^2}{\norm{A^+b}^2}
    = \frac{(\gamma_{A,b}^2-1)\norm{AA^+b}^2}{\norm{A^+b}^2}
	\le (\gamma_{A,b}^2-1)\norm{A}^2
\]
so from \eqref{eq dt*}, we have
\[
 \norm{\tx - x^*}^2 \le \norm{A^+}^2 4\eps_p\norm{\xi^*}^2 \le \norm{A^+}^2 4\eps_p\norm{x^*}^2  (\gamma_{A,b}^2-1)\norm{A}^2,
\]
and the second inequality of \eqref{eq p->x} follows, using the definition of $\kappa(A)$.
\end{proof}

\section{Dynamic Data Structures for Ridge Regression}
\label{sec:ridge_proof}

The data structure of Lemma~\ref{lem weighted tree} is simply a complete binary tree with $\ell$ leaves, each leaf
with weight $u_i$, and each internal node with weight equal to the sum of the weights of its children.
Sampling is done by walking down from the root, choosing left or right children with probability
proportional to its weight. Insertion and deletion are done by inserting or deleting the leaf $z$ that preserves the complete binary tree property, and updating the weights of its ancestors; in the case of deletion, first the leaf weight to be deleted is swapped with that of $z$, updating weights of ancestors. We also refer the reader to a more detailed description in~\cite{tang2019quantum, gst20}.
% \klcin{Do those references actually give a more detailed description?}

\begin{proof}[Proof of Lemma~\ref{lem dynsamp}]
Use Lemma~\ref{lem weighted tree} for the first part.
For the second, with a matrix $S$,
construct the data structure of Lemma~\ref{lem weighted tree}
for the row lengths of $SA$, in $O(m\lognd)$ time. To sample,
pick $i^*\in [m]$ with probability $\norm{(SA)_{i, *}}^2/\norm{SA}_F^2$,
using the newly constructed data structure. Then 
pick $j\in [\colms]$ with probability $(SA)_{i^*j}^2/\norm{(SA)_{i^*, *}}^2$.
Adding the probabilities across the choices of $i^*$,
the probability of choosing index $j$ is $\norm{(SA)_{*,j}}^2/\norm{SA}_F^2$,
as claimed.

Once a column is chosen, the time to determine the corresponding column length $\norm{(SA)_{*,j}}$
 is $O(m\lognd)$,
finding each $(SA)_{ij}$ for $i\in [m]$ in $O(\lognd)$ time. 
\end{proof}

We first re-state our theorem for dynamic ridge regression, before giving its proof.

\begin{theorem}[Theorem \ref{thm dyn ridgereg inf} restated]\label{thm dyn lowrank}
Given  matrices $A\in\R^{\rows\times\colms}, B\in\R^{\rows\times d'}$ and $\lambda>0$, let $X^* = \argmin_X \norm{AX-B}_F^2 + \lambda \norm{X}_F^2$. Let $\psi_\lambda = \norm{A}_F^2/(\lambda + \sigma_\rA^2)$, $\kappa_\lambda = (\lambda + \sigma_1(A)^2)/(\lambda + \sigma_\rA(A)^2)$ and 
$\hkap= \sqrt{(\lambda + \hsig_1^2)/(\lambda + \hsig_\rA^2)}$, where $\hsig_1$ and $\hsig_\rA$ are over and under estimates of $\sigma_1$ and $\sigma_\rA$ respectively.
Then, there exists a data structure that maintains $Y \in\R^{\colms\times d'}$ such that with probability at least $99/100$, 
\begin{equation*}
    \norm{Y - X^*}_F
    \le \left(\eps + 2\gamma \eps \right) \norm{X^*}_F  + \frac{\eps}{\sqrt{\lambda}}\norm{U_{k,\perp}B}_F,
\end{equation*}
where $\gamma^2 = \frac{\norm{B}_F^2}{\norm{A A^+ X^*}_F^2}$.
% Here $U_{\lambda,\perp}B$
% is the projection of $B$ onto the bottom
%  $m_S-p$ left singular vectors of $SA$, with $\lambda$ between
%  $\sigma_{p+1}(SA)$ and $\sigma_p(SA)$.
 Further, an entry $Y_{ij}$ for given $i,j$ can be computed in $
 O(m_S\lognd) = O(\eps^{-2}\hkap^2\psi_\lambda(\lognd)^2)
$ time.
 The time taken to compute $Y$ is $
\tO(\colms'\eps^{-4} \hkap^2 \psi_\lambda^2 \kappa_\lambda \log(\colms)).$

% \textcolor{red}{N: Changing theorem statement based on email. Also, we state this theorem twice in the main body, this does not seem necessary. }

%  Suppose $\DynSamp(A)$ is maintained
%  for $A\in\R^{\rows\times\colms}$.
%  For a given invocation of Algorithm~\ref{alg:dyn ridgereg},
%  with small constant failure probability, 
%  $\norm{A^\top S^\top \tX - X^*}_F
%     \le \eps(\norm{X^*}_F \gamma^2 + \frac1{\sqrt{\lambda}}\norm{U_{\lambda,\perp}B}_F)$.
%  Here $X^* = \argmin_X \norm{AX-B}_F^2 + \lambda \norm{X}_F^2$;
%  $\gamma_{\Al,\hB}$ given in Lemma~\ref{lem p->x} and after,
%  with \[{1+\gamma_{\hAl,\hB}^2 = };\]
%  and $\norm{U_{\lambda,\perp}B}$ is a projection of $B$ onto the bottom
%  $m_S-p$ left singular matrix of $SA$, with $\lambda$ between
%  $\sigma_{p+1}(SA)$ and $\sigma_p(SA)$.
%  An entry $(A^\top S^\top \tX)_{ij}$ for given $i,j$ can be computed in $
%  O(m_S\lognd) = O(\eps^{-2}\hkap^2\psi_\lambda(\lognd)^2)
% $ time.
%  The time taken to compute $\tX$ is $
% \tO(\colms'\eps^{-4} \hkap^2 \psi_\lambda^2 \kappa_\lambda \log(\colms)).$
% Here
% $\psi_\lambda = \norm{A}_F^2/(\lambda + \hsig_\rA^2)$,
% $\hkap= \sqrt{(\lambda + \hsig_1^2)/(\lambda + \hsig_\rA^2)}$,
% and $\kappa_\lambda = (\lambda + \sigma_1(A)^2)/(\lambda + \sigma_\rA(A)^2)$.
\end{theorem}

\begin{proof} %[Proof of Theorem \ref{thm dyn lowrank}]
Let 
\[
X_1 = \argmin_{X\in\R^{\colms\times d'} }\norm{SAX-SB}_F^2 + \lambda \norm{X}_F^2.
\]
We first show that
\begin{equation}\label{eq x1close}
 \norm{X_1 - X^*}_F \le \eps\gamma_{\Al, \hB}\norm{X^*}_F,
\end{equation}
which follows from Lemma~\ref{lem p->x}, applied to $\Al$ and $\hB$,
after showing that, for $\eps_p = \eps^2/\hkap^2$,
$X_1$ satisfies
% \nadiia{the proof of this theorem uses a static result for ridge regression. }
% \klcin{It's not a static (algorithmic) result, it's just a general fact. I uncommented the lemma, but haven't checked that the lemma makes sense in its current context. If it's needed in both papers, we'll need to deal with that. It's not presented as even original, it's more-or-less "for completeness"}
\begin{equation}\label{eq Sworks}
\norm{AX_1-B}_F^2 + \lambda \norm{X_1}_F^2 \le   (1+\eps_p/4)\Delta_*, \textrm{ where } \Delta_* = \norm{AX^*-B}_F^2 + \lambda \norm{X^*}_F^2,
\end{equation}
which in turn follows from Lemma~17 of \cite{ridge_reg}. That lemma considers a matrix $U_1$, comprising the
first $\rows$ rows of the left singular matrix of  $\hAl= \twomat{A}{\sqrt{\lambda} \Iden_\colms}$,
noting that the ridge objective can be expressed as $\min_X \norm{\hAl X - \twomat{B}{\mathbf{0}}}_F^2$.
The matrix $U_1=U\Sigma D$ in our terminology, as in Lemma~\ref{lem Al}, so the observations
of that lemma apply.

Lemma~17 of \cite{ridge_reg} requires that $S$ satisfies
\begin{equation}\label{eq embed mr}
\norm{U_1^\top S^\top S U_1 - U_1^\top U_1}\le 1/4,
\end{equation}
and
\begin{equation}\label{eq prod mr}
\norm{U_1^\top S^\top S(B-AX^*) - U_1^\top (B-AX^*)}_F \le \sqrt{\eps_p \Delta_*}.
\end{equation}
We have $\norm{\Al^+}^2 = 1/(\lambda + 1/\norm{A^+}^2) \le Z_\lambda$.
With the given call to \textsc{LenSqSample} to construct $S$, the number of rows
sampled is $m_S = O(\eps_p^{-1}Z_\lambda^2 \norm{A}_F^2\log(\colms))$,
% (ignoring the $Z/Z_\lambda$ term),
so the expected number of times
that row $i$ of $A$ is sampled is, up to a factor of $O(\log\colms)$,
\[
\eps_p^{-1}Z_\lambda^2 \norm{A}_F^2 \frac{\norm{A_{i,*}}^2}{\norm{A}_F^2}
	= \eps_p^{-1}Z_\lambda^2 \norm{A_{i,*}}^2
	\ge \eps_p^{-1} \norm{(U_1)_{i,*}}^2
	= \eps_p^{-1}  \norm{U_1}_F^2 \frac{\norm{(U_1)_{i,*}}^2}{\norm{U_1}_F^2},
\]
and so row $i$ is sampled at least the expected number of times it would be sampled under
$\eps_p^{-1}  \norm{U_1}_F^2\log\colms$ rounds of  length-squared sampling of $U_1$.
As shown by Rudelson and Vershynin (\cite{rv07}, see also \cite{kannan2017randomized}, Theorem 4.4),
this suffices to have, with high probability, a bound on the normed expression in \eqref{eq embed mr}
of
\[
\frac{\norm{U_1}\norm{U_1}_F}{\sqrt{\eps_p^{-1}  \norm{U_1}_F^2}} = \sqrt{\eps_p}\norm{U_1}\le \sqrt{\eps_p},
\]
so by adjusting constant factors in sample size, \eqref{eq embed mr} holds, for small enough $\eps_p$.

% The term $\eps_0 Z/Z_\lambda$ serves to ensure that $m_S\ge \eps_0^2 Z^2\norm{A}_F^2$,
% which implies that $S$ is with high probability a subspace $\eps_0$-embedding;
% this implies that $\norm{(SA)^+}=(1\pm\eps_0)\norm{A^+}$
% and that $\norm{SA}_F = (1\pm\eps_0)\norm{A}_F$, so that the terms from $A$
% can be used for $SA$.

To show that \eqref{eq prod mr} holds, we use the discussion of the basic matrix multiplication
algorithm discussed in \cite{kannan2017randomized}, Section 2.1, which implies that
\[
\E[\norm{U_1^\top S^\top S(B-AX^*) - U_1^\top (B-AX^*)}_F^2] \le \frac{\norm{U_1}_F^2\norm{B-AX^*}_F^2}{s}
\]
where $s$ is the number of length-squared samples of $U_1$.
Here $s=\eps_p^{-1}  \norm{U_1}_F^2\log\colms$, so \eqref{eq prod mr} follows
with constant probability by Chebyshev's inequality, noting that
$\norm{B-AX^*}_F\le \sqrt{\Delta_*}$.

Thus \eqref{eq embed mr} and  \eqref{eq prod mr} hold, so that by Lemma~17 of \cite{ridge_reg},
\eqref{eq Sworks} holds. We now apply Lemma~\ref{lem p->x}, which with \eqref{eq Sworks} and
$\eps_p = \eps^2/\hkap^2$, implies \eqref{eq x1close}.

Next we show that the (implicit) returned solution is close to the solution of \eqref{eq Sworks},
that is, \begin{equation}\label{eq Rworks}
\norm{A^\top S^\top \tX - X_1}_F^2 \le \eps \norm{X_1}_F^2.
\end{equation}
This is implied by Theorem~2 of \cite{chowdhury2018iterative}, since $A^\top S^\top \tX$ is
the output for $t=1$ of their Algorithm~1. (Or rather, it is their output for each column of $\tX$ and
corresponding column of $B$.) To invoke their Theorem~2, we need to show that their equation~(8)
holds, which per their discussion following Theorem~3, holds for ridge leverage score
sampling, with $O(\eps^{-2}\sd_\lambda\log\sd_\lambda)$ samples,
which our given $m_R$ yields.

% Here we need
% $\eps_1 \le \eps(1+\sqrt{2\lambda}/\norm{SA})$. Theorem~2 of \cite{chowdhury2018iterative}
% involves a term $U_{k,\perp}^\top B$, a projection onto a subspace of $\colspan(SA)$;
% it suffices for us that
% \[
% \norm{U_{k,\perp}^\top B}_F \le \norm{(SA)(SA)^+B}_F\le \norm{SA}\norm{(SA)^+B},
% \]
% and we can note $E[\norm{SA}_F^2] = \norm{A}_F^2$ and Markov's inequality to
% obtain $\norm{SA}\le\norm{SA}_F\le c\norm{A}_F$, for constant $c>0$. This implies
% $norm{U_{k,\perp}^\top B}_F \le c\norm{A}_F\norm{}$

When we invoke their Theorem~2,
we obtain
\begin{align}\label{eq drin}
\norm{A^\top S^\top \tX - X_1}_F
        & \le \eps (\norm{X_1}_F + \frac1{\sqrt{\lambda}}\norm{U_{k,\perp}B}_F).
\end{align}

Combining with \eqref{eq x1close} and using the triangle inequality,
we have that, abbreviating $\gamma_{\hAl,\hB}$ as $\gamma$,
up to the additive $U_{k,\perp}$ term in \eqref{eq drin}, we have
% \begin{align*}
% \norm{A^\top S^\top \tX - X^*}_F
%  	   & \le \norm{A^\top S^\top \tX - X_1}_F + \norm{X_1 - X^*}_F
% 	\\ & \le \eps \norm{X_1}_F + \eps(\gamma^2-1)\norm{X^*}_F
% 	\\ & \le \eps\left[\norm{X^*}_F + \eps(\gamma^2-1)\norm{X^*}_F +  \gamma\norm{X^*}_F\right]
% 	\\ & = \eps\norm{X^*}_F [1 + (\gamma^2-1)(1+\eps)]
% 	\\ & \le \eps\norm{X^*}_F 2\gamma^2,
% \end{align*}
% \textcolor{red}{N: editing the above:}
\begin{align*}
\norm{A^\top S^\top \tX - X^*}_F
 	   & \le \norm{A^\top S^\top \tX - X_1}_F + \norm{X_1 - X^*}_F
	\\ & \le \eps \norm{X_1}_F + \norm{X_1 - X^*}_F + \frac{\eps}{\sqrt{\lambda}}\norm{U_{k,\perp}B}_F
	\\ & \le \eps \norm{X_*}_F + (1+\eps)\norm{X_1 - X^*}_F + \frac{\eps}{\sqrt{\lambda}}\norm{U_{k,\perp}B}_F
	\\ & \le \eps\norm{X^*}_F + 2 \eps \gamma \norm{X_*}_F + \frac{\eps}{\sqrt{\lambda}}\norm{U_{k,\perp}B}_F
	\\ & \le \eps\left(1 + 2\gamma \right) \norm{X^*}_F + \frac{\eps}{\sqrt{\lambda}}\norm{U_{k,\perp}B}_F
\end{align*}

for small enough $\eps$,
% and then finally,
% \[
% \norm{A^\top S^\top \tX - X^*}_F
%     \le \eps\left(1 + 2\gamma \right) \norm{X^*}_F  + \frac{\eps}{\sqrt{\lambda}}\norm{U_{k,\perp}B}_F,
% \]
as claimed.

The time is dominated by that for computing
$\tA^{-1} SB$, where $\tA = SARR^\top A^\top S^\top$,
which we do via the conjugate gradient method. Via standard
results (see, e.g., \cite{vishnoi}, Thm 1.1), 
in $O((T+m_S)\sqrt{\kappa(\tA)}\log(1/\alpha)) \colms'$
time, where $T$ is the time to compute the product of $\tA$ with a vector,
we can obtain $\tX$ with $\norm{\tX - \tA^{-1} SB}_\tA\le \alpha\norm{\tA^{-1} SB}_\tA$,
where the $\tA$-norm is $\norm{x}_\tA = x^\top\tA x$.
Since $S$ and $R$ are (at least) constant-factor subspace embeddings,
the singular values of $SAR$ are within a constant factor of those of $A$,
and so $\kappa(\tA)$ is within a constant factor of
\[
\kappa(AA^\top + \lambda\Iden)= (\lambda + \sigma_1(A)^2)/(\lambda + \sigma_1(A)^2) = \kappa_\lambda^2.
\]
We have
\begin{align*}
T   & = O(m_R m_S)
    = \tO(
        \eps^{-2}\log m_S Z_\lambda^2 \norm{A}_F^2
        \eps^{-2}\hkap^2 Z_\lambda^2 \norm{A}_F^2\log(\colms)
        )
    \\ & = \tO(\eps^{-4} \hkap^2 Z_\lambda^4 \norm{A}_F^4 \log(\colms))
\end{align*}
Our running time is $\tO(T\kappa_\lambda\log(1/\eps)) \colms'$, with $T$
as above. Translating to the notation using $\psi$ terms, the result follows.
% The time is dominated by that for computing $SAR(SAR)^\top$ and solving
% for $\tX$. Here $m_S = O((\eps^{-2}\hkap^2 + Z^2/Z_\lambda^2) Z_\lambda^2 \norm{A}_F^2\log(\colms))$
% and $m_R$ is, from \textsc{LenSqEmbed}, Alg.~\ref{alg:lensqembed},
% \[
% O(\frac{\hat m_R}{\hat k}Z^2 \norm{A}_F^2)
%     = O(\frac{\eps^{-2}m_S\log m_S}{m_S}Z^2 \norm{A}_F^2)
%     = O(\eps^{-2}(\log m_S)Z^2 \norm{A}_F^2), 
% \]
% Computing $SAR(SAR)^\top$ needs time $O(m_S^2m_R)$,
% while solving an $m_S\times m_S$ system takes time
% $O(m_S^3)$, so overall, ignoring the $Z^2/Z_\lambda^2$ terms,
% \begin{align*}
% O(m_S^2(m_S+m_R))
%       & = O(\eps^{-4}\hkap^4 Z_\lambda^4 \norm{A}_F^4\log^2(\colms)
%         (\eps^{-2}\hkap^2Z_\lambda^2 \norm{A}_F^2\log(\colms) + \eps^{-2}(\log m_S)Z^2 \norm{A}_F^2  ))
%     \\ & = O(\eps^{-6}\hkap^4 Z_\lambda^4 \norm{A}_F^6\log^2(\colms)
%         (\hkap^2 Z_\lambda^2 \log(\colms) + Z^2\log (m_S)  ))
% \end{align*}
% time is needed. The $Z/Z_\lambda$ term contributes (we can ignore dominated cross-terms)
% \begin{align*}
%  O(Z^6 \norm{A}_F^6\log^2(\colms)
%         (\log(\colms) + \eps^{-2}(\log m_S)))^3.
% \end{align*}
\end{proof}

\section{Sampling from a Low-Rank Approximation}
\label{sec:LRA_proof}

We need the following lemma, implied by the algorithm and analysis
in Section 5.2 of \cite{boutsidis2016optimal}; for completeness we include a proof.

\begin{lemma}\label{lem PCP}
If $S\in\R^{m_S\times \rows}$ and $R$ are such that $SA$ is a PCP of $A\in\R^{\rows\times\colms}$,
and $SAR$ is a PCP of $SA$, for error $\eps$ and rank $\rkt$,
and $U\in\R^{m_S\times k}$ has orthonormal columns
such that
$\norm{(\Iden-UU^\top)SAR}_F \le (1+\eps)\norm{SAR - [SAR]_\rkt}$,
then 
\[
Y^* = \argmin_Y \norm{YU^\top SA - A}_F 
\]
has
\begin{equation}\label{eq pcp}
    \norm{Y^* U^\top SA - A}_F\le (1+O(\eps))\norm{A-A_{\rkt}}_F.
\end{equation}
We also have
\[
\norm{U^\top SA}_F^2 \ge \norm{A_\rkt}_F^2 - O(\eps)\norm{A}_F^2.
\]
\end{lemma}

\begin{proof}
Note that for matrix $Y$,
$Y(\Iden - Y_\rkt^+Y_\rkt) = (\Iden - Y_\rkt Y_\rkt^+) Y$,
and that $UU^\top SA$ is no closer to $SA$ than
is the projection of $SA$ to the rowspace of $U^\top SA$,
and that $UU^\top = (SAR)_{\rkt}(SAR)_{\rkt}^+$
we have
\begin{align*}
\norm{A - Y^*U^\top SA}_F
        = \norm{A(\Iden - (U^\top SA)^+U^\top SA)}_F
     & \le (1+\eps)\norm{SA(\Iden - (U^\top SA)^+U^\top SA)}_F
    \\ & \le (1+\eps)\norm{(\Iden - UU^\top)SA}_F
    \\ & \le (1+\eps)^2\norm{(\Iden - UU^\top)SAR}_F
    \\ & \le (1+\eps)^3\norm{(\Iden - (SAR)_{\rkt}(SAR)_{\rkt}^+)SAR}_F
    \\ & \le (1+\eps)^3\norm{(\Iden - (SA)_{\rkt}(SA)_{\rkt}^+ )SAR}_F
    \\ & \le (1+\eps)^4 \norm{(\Iden - (SA)_{\rkt}(SA)_{\rkt}^+)SA}_F
    \\ & = (1+\eps)^4 \norm{SA(\Iden - (SA)_{\rkt}^+(SA)_{\rkt})}_F
    \\ & \le (1+\eps)^4 \norm{SA(\Iden - A_{\rkt}^+ A_{\rkt})}_F
    \\ & \le (1+\eps)^5 \norm{A(\Iden - A_{\rkt}^+ A_{\rkt})}_F
    \\ & = (1+\eps)^5 \norm{A-A_{\rkt}}_F = (1+O(\eps)) \norm{A-A_{\rkt}}_F,
\end{align*}
as claimed.

For the last statement:
we have $\norm{SA}_F^2 \ge (1-\eps)\norm{A}_F^2$,
since $SA$ is a PCP, and by considering the projection of $A$ onto the
rowspans of blocks of $\rkt$ of
its rows. We have also $\norm{SA-[SA]_\rkt}_F^2 \le (1+\eps)\norm{A-[A]_\rkt}^2$,
using that $SA$ is a PCP. Using these observations, we have
\begin{align*}
\norm{[SA]_\rkt}_F^2
       & = \norm{SA}_F^2 - \norm{SA-[SA]_\rkt}_F^2
    \\ & \ge (1-\eps)\norm{A}_F^2 - (1+\eps) \norm{A-[A]_\rkt}_F^2
    \\ & = \norm{[A]_\rkt}_F^2 - \eps(\norm{A}_F^2 + \norm{A-[A]_\rkt}_F^2)
    \\ & \ge \norm{[A]_\rkt}_F^2 - 3\eps\norm{A}_F^2.
\end{align*}
Similarly, $\norm{[SAR]_\rkt}_F^2 \ge \norm{[SA]_\rkt}_F^2 - 3\eps\norm{SA}_F^2$,
using that $SAR$ is a PCP of $SA$.
We then have, using these inequalities, the PCP properties, and the hypothesis for $U$, that
\begin{align*}
\norm{U^\top SA}_F^2
       & = \norm{UU^\top SA}_F^2
    \\ & = \norm{SA}_F^2 - \norm{(I-UU^\top)SA}_F^2
    \\ & \ge (1-\eps)\norm{SAR}_F^2 - (1+\eps)^2\norm{SAR - [SAR]_\rkt}_F^2
    \\ & \ge \norm{[SAR]_\rkt}_F^2 - 4\eps\norm{SAR}_F^2
    \\ & \ge (\norm{[SA]_\rkt}_F^2 - 3\eps\norm{SA}_F^2) - 4(1+\eps)\eps\norm{SA}_F^2
    \\ & \ge (\norm{[A]_\rkt}_F^2 - 3\eps\norm{A}_F^2) - 3\eps(1+\eps)\norm{A}_F^2 - 4(1+\eps)^2\eps\norm{A}_F^2
    \\ & \ge \norm{[A]_\rkt}_F^2 - 13\eps\norm{A}_F^2,
\end{align*}
for small enough $\eps$,
and the last statement of the lemma follows.
\end{proof}

%\klcin{could find conditioner for $AR$ to accelerate sampling}

Before a proof, we give a re-statement of Theorem~\ref{thm lra inf}.

\begin{theorem}[Dynamic Data Structure for LRA]\label{thm dyn LRA}
Given a matrix  $A\in\R^{\rows\times \colms}$,
 target rank $k$, and 
 estimate 
%$\hsig_\rA\le 1/\norm{A^+}$,
 $\hsig_{k} \le \sigma_{k}(A)$,
error parameter $\eps>0$, and 
estimate $\tau$ of $\norm{A-A_{k}}_F^2$, there exists a data structure representing a matrix $Z \in \R^{\rows\times \colms}$ with rank $k$ such that if $\norm{A_{k}}_F^2\ge \eps\norm{A}_F^2$, with probability at least $99/100$, 
 $\norm{A-Z}^2_F \le (1+O(\eps))\norm{A-A_{k}}^2_F$, where $A_{k}$ is the best rank-$k$ approximation
to $A$. Further, the time taken to construct the representation of $Z$ is 
\[
\tO(\eps^{-6}k^3 + \eps^{-4} \psi_\lambda (\psi_\lambda + k^2 + k \psi_{k})),
\]
where $\psi_\lambda= \norm{A}_F^2/(\tau/k + \hsig_\rA^2)$
and $\psi_{k} = \norm{A}_F^2/\sigma_{k}(A)$.
Given $j\in [\colms]$, $i\in [\rows]$ can be generated with probability
$(Z)_{ij}^2/\norm{Z)_{*,j}}^2$ in expected time $O(\norm{A}_F^2/\hsig_\rA^2 + m_R^2 \kappa^2)$, where $\kappa$ is the condition number of $A$,
and $m_R = O(\rA\log\rA + \eps^{-1}\rA)$.
\end{theorem}

\begin{proof} %[Proof of Theorem \ref{thm lra inf}]
The matrix $Z$ is the implicit output of Algorithm~\ref{alg:dyn lowrank}.
In that algorithm, the choice of $m_{S} = O(\hat m_{\SS } Z_\lambda^2 \norm{A}_F^2)$ rows constitutes
an effective $k\hat m_{S}= O(\eps^{-2}k\log k)$ ridge-leverage score
samples of the rows of $A$.
We assume that the input $\tau$ is within a constant factor 
of $\norm{A-A_{k}}^2_F$, so that $\lambda = \tau/k$
is within a constant factor of $\norm{A-A_k}_F^2/k$.
%\nadiia{why is this true?}
Theorem~6 of \cite{cohen2017input}
implies that under these conditions,
${\SS }A$ will be a rank-$k$ Projection-Cost Preserving (PCP) sketch of $A$ with error parameter $\eps$,
a $(k,\eps)$-PCP.
%\klcin{Will this still work if $\tau$ is e.g. an under-estimate?}
%\nadiia{if we scale up $\tau$ to always be at least $\|A -A_{k} \|^2_F$ we should be able to over sample more and still have correct guarantee}
% The appended $S''$ is of a size to ensure
% that with high probability, ${\SS }$ is a subspace $\eps_0$-embedding,
% so that in particular $\norm{(\SS A)^+} \approx \norm{A^+}$,
% and $Z, Z_\lambda$ are appropriate parameters to specify the
% samplers $R_1$ and $R_3$.

Similarly to $\SS $, $R_1$ will be a (column) rank-$k$ PCP of ${\SS }A$,
here using that the PCP properties of $SA$ imply that
$\norm{(S(A-A_{k})}_F^2 = (1\pm\eps)\norm{A-A_{k}}_F^2$, and
so the appropriate $\lambda$, and $Z_\lambda$, for $SA$ are within constant factors
of those for $A$. 
%\klcin{assuming here that min sigma of SA is unimportant}
%\nadiia{in a proof of the pcp this inequality requires a scalar chernoff bound argument, I think. we may not need to prove this if we get that we have overestimates of ridge leverage scores. is that what you mean by min sigma is unimportant?}
Let $\tA= \SS AR_1$.
Lemma~16 and Theorem~1 of \cite{cohen2017input} imply
that applying their Algorithm~1 to $\tA$ yields
$S_2\in\R^{m_{S_2}\times m_{\SS }}$ so that $S_2\tA $
is a $(k,\eps)$-PCP for $\tA $, and similarly
$S_2\tA R_2$ is a $(k,\eps)$-PCP for $S_2\tA $.

We apply Lemma~\ref{lem PCP} with $\tA^\top$, $R_2^\top$,
$S_2^\top$, and $V^\top$ in the roles of $A$, $S$, $R$, and $U$ in the lemma.
We obtain that $\tY = (\tA R_2 V)^+\tA = \argmin_Y \norm{\tA R_2 V Y - \tA }_F $ has
$\norm{\tA R_2 V \tY - \tA }_F
\le (1+O(\eps))\norm{\tA  - \tA_{k}}_F$,
that is, $U$ as constructed in Algorithm~\ref{alg:dyn lowrank}
has $UU^\top \tA = \tA R_2 V \tY$, and
therefore satisfies the conditions of
Lemma~\ref{lem PCP} for $A$, $\SS $, $R_1$.
This implies that
$Y^* = A(U^\top SA)^+=  \argmin_Y \norm{YU^\top SA - A}_F  $ has
\begin{equation}\label{eq pcp2}
    \norm{Y^* U^\top SA - A}_F\le (1+O(\eps))\norm{A-A_{k}}_F.
\end{equation}
It remains to solve
the multiple-response regression problem
$min_Y \norm{YU^\top SA - A}_F$,
which we do more quickly using the samplers $R_3$ and $R_4$.

We next show that $R_3^\top$ is a subspace
$\eps_0$-embedding of $(U^\top SA)^\top$,
and supports low-error matrix product estimation,
so that Thm.~36 of \cite{CW13} can be applied.
Per Lemma~\ref{lem embed} and per Lemma~32 of \cite{CW13},
$k \hat m_{R_3} = O(\eps_0^{-2}k\log k + \eps^{-1}k)$
leverage-score samples of the columns of $U^\top SA$
suffice for these conditions to hold.

To obtain $k \hat m_{R_3}$ leverage score samples,
we show that $1/\eps$ length-squared samples of
the columns of $SA$ suffice to contain one length-squared 
sample of $U^\top SA$, and also
that
$\norm{(U^\top SA)^+}\le 1/\hsig_{k}$, using the input condition
on $\hsig_{k}$ that $\sigma_{k}(A)\ge \hsig_{k}$,
so that the given value of $m_{R_3}$ in the call to $\textsc{LenSqSample}$
for $R_3$ is valid.

For the first claim, by hypothesis
$\norm{A_{k}}_F^2\ge \eps\norm{A}_F^2$,
and by adjusting constants, $U$ as computed satisfies
the conditions of Lemma~\ref{lem PCP} for some $\eps'=\alpha\eps$
for constant $\alpha>0$, so by that lemma and by hypothesis
\begin{align*}
\norm{U^\top SA}_F^2
        & \ge \norm{A_{k}} - O(\alpha\eps)\norm{A}_F^2 
    \\ & \ge \eps(1-O(\alpha))\norm{A}_F^2
    \\ & \ge \eps(1-O(\alpha))(1-\eps)\norm{SA}_F^2,
\end{align*}
so adjusting constants, we have $\norm{U^\top SA}_F^2 \ge \eps\norm{SA}_F^2$.
Using that $U$ has orthonormal columns, we have for $j\in[\colms]$
that $\norm{U^\top SA_{*,j}}/\norm{U^\top SA}_F^2 \le \norm{SA_{*,j}}/\eps\norm{SA}_F^2$, so the probability
of sampling $j$ using length-squared probabilities for
$SA$ is least $\eps$ times that for $U^\top SA$.

For the claim for the value of $m_{R_3}$ used for $R_3$, using the
PCP properties of $SA$ and $SAR_1$, we have
\[
\sigma_{k} (U^\top SA) = \sigma_{k} (UU^\top SAR_1) = \sigma_{k}(SAR_1) \ge (1-\eps) \sigma_{k}(SA) \ge (1-O(\eps)) \sigma_{k}(A).
\]
so the number of length-squared samples returned by \textsc{LenSqSample} suffice.

So using Thm.~36 of \cite{CW13},
$\tY_3 = \argmin_Y \norm{(Y U^\top SA - A)R_3}_F$
satisfies
\[
\norm{\tY_3 U^\top SA - A}_F
    \le (1+\eps)\min_Y \norm{Y U^\top SA - A}_F
    \le (1+O(\eps))\norm{A-A_k}_F,
\]
where the last inequality follows from \eqref{eq pcp}.

Similar conditions and results can be applied
to direct leverage-score sampling of the columns
of $U^\top SAR_3$, resulting in
$\tY_4 = \min_Y \norm{(Y U^\top SAR_3 - AR_3)R_4}_F$,
where there is $m_{R_4} = O(\eps_0^{-2}k\log k + \eps^{-1}k)$ such that these
conditions hold for $R_4$.
This implies $\tY_4$ is an approximate solution
to $\min_Y \norm{(YU^\top SA-A)R_3}_F$,
and therefore
$AR\tY_4 = AR(U^\top SAR)^+$
has $\norm{AR\tY_4 U^\top SA - A}_F \le (1+O(\eps))\norm{A-A_k}_F$,
as claimed. We have $W \gets (U^\top SAR)^+ U^\top = \tY_4 U^\top $,
so the claimed output condition on $ARW\SS A$ holds.

Turning to the time needed,
Lemma~16 and Theorem~1 of \cite{cohen2017input} imply
that the time needed to construct
$S_2$ and $R_2$ is
\begin{equation}\label{eq LR S_2}
    O(m_{R_1}m_{\SS } + k^2 m_{\SS }) = O(m_S(m_S + k^2))
\end{equation}

The time needed to construct $V$ from $S_2SAR_1R_2\in\R^{m_{S_2}\times m_{S_2}}$ is
\begin{equation}\label{eq LR V}
O(m_{S_2}^3) = \tO(\eps^{-6}k^3)
\end{equation}

The time needed to construct $U$ from $V\in\R^{m_{R_2}\times k}$
and $\SS AR_1R_2\in \R^{m_{\SS}\times m_{R_2}}$
by multiplication and QR factorization is
\begin{equation}\label{eq LR U}
    O(k m_{R_2} m_{\SS } + k^2 m_{\SS }) = \tO(m_S k^2 \eps^{-2}  )
\end{equation}

Computation of $U^\top SAR_3$ requires
$O(k m_S m_{R_3})$ time, where
$m_{R_3} = O(\hat m_{R_3}\eps^{-1} Z_{k}^2\norm{A}_F^2$,
and $\hat m_{R_3} = O(\log k + \eps^{-1})$,
that is,
\begin{equation}\label{eq LR R3}
    \tO(m_S k \eps^{-2} Z_{k}^2\norm{A}_F^2)
\end{equation}
time.
Leverage-score sampling of the rows of
$(U^\top SAR_3)^\top \in\R^{m_{R_3}\times k}$
takes, applying Theorem~\ref{thm levsample}
% \nadiia{this analysis uses the leverage score sampling data structure.}
% \klcin{there's now a theorem statement for the lev score sampling running time.}
and using $m_{R_4} = O(k(\log k + \eps^{-1}))$,
time at most
\begin{align}\label{eq LR LS}
O(log(m_{R_3})k & m_{R_3} + k^\omega\log\log(k m_{R_3}) + k^2\log m_{R_3} +  m_{R_4} k m_{R_3}^{1/log(m_{R_3})})
\\ & = \tO(k\eps^{-2}Z_{k}^2\norm{A}_F^2  + k^\omega +  k^2 \eps^{-1})
\end{align}
Computation of $(U^\top SAR)^+$ from $U^\top SAR$
requires $O(k^2 m_{R_4}) = \tO(\eps^{-1}k^3)$
time. (With notation that $R$ has $m_R = m_{R_4}$ columns.)
Given $(U^\top SAR)^+$, computation of $W=(U^\top SAR)^+U^\top$ requires
\begin{equation}\label{eq LR U+}
    O(k m_R m_S) = \tO(m_S k^2 \eps^{-1})
\end{equation}
 time.
Putting together
\eqref{eq LR S_2},\eqref{eq LR U}, \eqref{eq LR R3}, \eqref{eq LR U+},
\eqref{eq LR V}, \eqref{eq LR LS}, we have
\begin{align*}
    & O(m_S(m_S + k^2))
    + \tO(m_S k^2 \eps^{-2})
    + \tO(m_S k \eps^{-2} Z_{k}^2\norm{A}_F^2)
    + \tO(m_S k^2 \eps^{-1})
    \\ & + \tO(\eps^{-6}k^3)
         + tO(k\eps^{-2}Z_{k}^2\norm{A}_F^2  + k^\omega +  k^2 \eps^{-1})
    \\ & = \tO(m_S(m_S + k^2 \eps^{-2} + k \eps^{-2} Z_{k}^2\norm{A}_F^2)
            + \eps^{-6}k^3)
\end{align*}
Here
$m_S = O(\hat m_S Z_\lambda^2\norm{A}_F^2) 
= \tO(\eps^{-2} Z_\lambda^2\norm{A}_F^2)$,
so the time is
\[
\tO(\eps^{-6}k^3 + \eps^{-4} Z_\lambda^2\norm{A}_F^2 (
Z_\lambda^2\norm{A}_F^2 + k^2 + k Z_{k}^2\norm{A}_F^2))
\]

Queries as in the theorem statement for given $j\in [\colms]$ can be
answered as in \cite{tang2019quantum},
in the time given.
Briefly: given $j$, let $v\gets (WSA)_{*,j}\in \R^{m_R}$. Let $\hA$ denote $AR$.
Using $\textsc{LenSqSample}$ (and large $m_R$), generate
a sampling matrix $S_3$ with $O(Z^2\norm{\hA}^2)$ rows,
and estimate $\norm{\hA v}^2\approx \beta_v \gets \norm{S\hA v}^2$.
Generate $i\in[\rows]$ via rejection sampling as follows.
In a given trial, pick $j^*\in[\colms]$ with probability
proportional to
$\norm{\hA_{*,j}}^2 v_j^2$ using $\DynSamp(A)$,
then pick $i\in [m_R]$ with probability
$\hA_{i,j^*}^2/\norm{\hA_{*,j^*}}^2$.
This implies that $i\in[\rows]$ has been
picked with probability
$p_i= \sum_j \hA_{ij}^2v_j^2/\sum_j \norm{\hA_{*,j}}^2v_j^2$.
Now for a value $\alpha>0$,
accept with probability $q_i/\alpha p_i$,
where $q_i= (\hA_{i,*}v)^2/\beta_v$,
otherwise reject. This requires $\alpha \ge q_i/p_i$.
and takes expected trials $\alpha$. So an upper bound
is needed for
\[
\frac{q_i}{p_i}
    = \frac{(\hA_{i,*}v)^2}{\beta_v}\frac{\sum_j \norm{\hA_{*,j}}^2v_j^2}{\sum_j \hA_{ij}^2v_j^2}.
\]
We have $(\hA_{i,*}v)^2\le m_R \sum_j \hA_{ij}^2v_j^2$ using Cauchy-Schwarz, and $\beta_v \approx \norm{\hA v}^2 \ge \norm{v}^2/\norm{\hA^+}^2$, and also
$\sum_j \norm{\hA_{*,j}}^2v_j^2 \le \norm{v}^2 \max_j\norm{\hA_{*,j}}^2 \le \norm{v}^2\norm{\hA}^2$.
Putting this together $\alpha \ge m_R \norm{\hA}^2\norm{\hA^+}^2 = O(m_R\kappa(A))$
will do. The work per trial is $O(m_R\lognd)$
,and putting that together with the time to compute $\beta_v$, the theorem follows.
\end{proof}

\section{Additional Experiments}
\label{sec more exp}

\subsection{Low-Rank Approximation}
\label{sec:appendix_exp_lra}

As we stated in Section~\ref{exp:lra}, although the accuracy of our algorithm is slightly worse, by increasing the sample size slightly, our algorithm achieves a similar accuracy as~\cite{quantum_practice}, but still has a faster runtime. The results are shown in Table~\ref{tab:appendix_lra}. Here we set $(r,c) = (500, 800)$ for MovieLens 100K and $(r,c) = (700, 1100)$ for KOS data for our algorithms, but do not change $(r,c)$ for the algorithm in~\cite{quantum_practice} as in Section~\ref{exp:lra}.

\begin{table}[h]
\caption{Performance of our algorithm and ADBL on MovieLens 100K and KOS data, respectively.}%. The above is for MovieLen 100K and the below is for KOS dats}
\label{tab:appendix_lra}
\centering
\begin{tabular}{|c|c|c|c|}
\hline
& $k=10$ & $k=15$ & $k=20$ \\
\hline
$\eps$(Ours) & 0.0323 & 0.0439 & 0.0521\\
\hline
$\eps$(ADBL) & 0.0262 & 0.0424 & 0.0538\\
\hline
\hline
 Runtime  & \multirow{2}{*}{0.341s} & \multirow{2}{*}{0.365s} & \multirow{2}{*}{0.370s} \\
(Ours, Query) & & & \\
\hline
Runtime & \multirow{2}{*}{0.412s} & \multirow{2}{*}{0.417s} & \multirow{2}{*}{0.415s} \\
(Ours, Total) & & & \\
\hline
Runtime& \multirow{2}{*}{0.863s} & \multirow{2}{*}{0.917s} & \multirow{2}{*}{1.024s} \\
(ADBL, Query) & & & \\
\hline
Runtime & \multirow{2}{*}{0.968s} & \multirow{2}{*}{1.003s} & \multirow{2}{*}{1.099s} \\
(ADBL, Total)& & & \\
\hline
Runtime of SVD & \multicolumn{3}{c|}{2.500s}\\
\hline
\end{tabular}
\begin{tabular}{|c|c|c|c|}
\hline
& $k=10$ & $k=15$ & $k=20$ \\
\hline
$\eps$(Ours) & 0.0291 & 0.0390 & 0.0476\\
\hline
$\eps$(ADBL)  & 0.0186 & 0.0295 & 0.0350\\
\hline
\hline
 Runtime  & \multirow{2}{*}{0.826s} & \multirow{2}{*}{0.832s} & \multirow{2}{*}{0.831s} \\
(Ours, Query) & & & \\
\hline
Runtime & \multirow{2}{*}{0.979s} & \multirow{2}{*}{0.994s} & \multirow{2}{*}{0.986s} \\
(Ours, Total) & & & \\
\hline
Runtime& \multirow{2}{*}{1.501s} & \multirow{2}{*}{1.643s} & \multirow{2}{*}{1.580s} \\
(ADBL, Query) & & & \\
\hline
Runtime & \multirow{2}{*}{1.814s} & \multirow{2}{*}{1.958s} & \multirow{2}{*}{1.897s} \\
(ADBL, Total)& & & \\
\hline
Runtime of SVD & \multicolumn{3}{c|}{36.738s}\\
\hline
\end{tabular}
\end{table}

\vfill

% \end{document}

\end{document}